\documentclass[journal,10pt]{IEEEtran}
\usepackage{amsmath}
\usepackage{amssymb}
\usepackage{amsfonts}
\usepackage{graphicx}
\usepackage{epsfig}
\usepackage{psfrag}
\usepackage{tabulary}
\usepackage{booktabs}
\usepackage[skip=0pt]{caption}
\usepackage{pifont}
\usepackage{color}
\usepackage{bm}
\usepackage{subcaption}
\usepackage{url}
\usepackage[colorlinks,linkcolor=blue,anchorcolor=blue,citecolor=blue,urlcolor=black]{hyperref}
\usepackage{breakurl}
\usepackage{cite}
\usepackage{epstopdf}

\begin{document}
\title{Joint Energy and Spectrum Cooperation for Cellular Communication Systems}

\author{Yinghao Guo, Jie Xu, Lingjie Duan, and Rui Zhang
\thanks{ Part of this paper has been presented in the IEEE International Conference on Communications (ICC), Sydney, Australia, 10-14 June, 2014.}
\thanks{Y. Guo and J. Xu are with the Department of Electrical and Computer Engineering, National University of Singapore (e-mail: yinghao.guo@nus.edu.sg, elexjie@nus.edu.sg).}
\thanks{L. Duan is with the Engineering Systems and Design Pillar, Singapore University of Technology and Design (e-mail: lingjie\_duan@sutd.edu.sg).}
\thanks{R. Zhang is with the Department of Electrical and Computer Engineering, National University of Singapore (e-mail:elezhang@nus.edu.sg). He is also with the Institute for Infocomm Research, A*STAR, Singapore.}
\vspace{-2em}}

\maketitle

\IEEEpeerreviewmaketitle
\setlength{\baselineskip}{1\baselineskip}
\newtheorem{definition}{\underline{Definition}}[section]
\newtheorem{fact}{Fact}
\newtheorem{assumption}{Assumption}
\newtheorem{theorem}{\underline{Theorem}}[section]
\newtheorem{lemma}{\underline{Lemma}}[section]
\newtheorem{corollary}{\underline{Corollary}}[section]
\newtheorem{proposition}{\underline{Proposition}}[section]
\newtheorem{example}{\underline{Example}}[section]
\newtheorem{remark}{\underline{Remark}}[section]
\newtheorem{algorithm}{\underline{Algorithm}}[section]
\newtheorem{observation}{\underline{Observation}}[section]
\newcommand{\mv}[1]{\mbox{\boldmath{$ #1 $}}}
\newcommand{\cred}{\color{red}}
\begin{abstract}
Powered by renewable energy sources, cellular communication systems usually have different wireless traffic loads and available resources over time. To match their traffics, it is beneficial for two neighboring systems to cooperate in resource sharing when one is excessive in one resource (e.g., spectrum), while the other is sufficient in another (e.g., energy). In this paper, we propose a joint energy and spectrum cooperation scheme between different cellular systems to reduce their operational costs. When the two systems are fully cooperative in nature (e.g., belonging to the same entity), we formulate the cooperation problem as a convex optimization problem to minimize their weighted sum cost and obtain the optimal solution in closed form. We also study another partially cooperative scenario where the two systems have their own interests. We show that the two systems seek for partial cooperation as long as they find inter-system complementarity between the energy and spectrum resources. Under the partial cooperation conditions, we propose a distributed algorithm for the two systems to gradually and simultaneously reduce their costs from the non-cooperative benchmark to the Pareto optimum. This distributed algorithm also has proportional fair cost reduction by reducing each system's cost proportionally over iterations. Finally, we provide numerical results to validate the convergence of the distributed algorithm to the Pareto optimality and compare the centralized and distributed  cost reduction approaches for fully and partially cooperative scenarios.
\end{abstract}


\section{Introduction}

With the exponential increases of the wireless subscribers and data traffic in recent years, there has also been a tremendous increase in  the energy consumption of the cellular systems and energy cost constitutes a significant part of wireless system's operational cost \cite{Hasan2011}. To save operational costs, more and more cellular operators are considering to power their systems  with  renewable energy supply, such as solar and wind sources. For instance, Huawei has adopted  renewable energy solution  at cellular base stations (BSs) in Bangladesh \cite{Huawei}. Nevertheless, unlike the conventional energy from grid, renewable energy (e.g., harvested through solar panel or wind turbine) is intermittent in nature and can have different availabilities over time and space. Traditional methods like energy storage with the use of capacity-limited and expensive battery are far from enough for one system to manage the fluctuations.  To help mitigate such uncertain renewable energy fluctuations and  shortage, {\it energy cooperation} by sharing one system's excessive energy to the other has been proposed for cellular networks (see \cite{Chia2013,Xu2013}). However, one key problem that remains unaddressed yet to implement this cooperation is how to motivate one system to share its energy to the other system with some benefits in return (e.g., collecting some other resource from the other system).

Besides energy, spectrum is another important resource for the operation of a cellular system and the two resources can complement each other. For example, to match peak-hour wireless traffics with limited spectrum, one system can increase energy consumption for transmission at a high operational cost \cite{Hasan2011}.  As it is unlikely that the two neighboring systems face spectrum shortage at the same time, it is helpful for them to share spectrum. Note that the similar idea of spectrum cooperation can be found in the context of cognitive radio networks \cite{Goldsmith2009}. However, like energy cooperation, spectrum cooperation here also faces the  problem of  how to motivate one system to share spectrum with the other.\footnote{It is possible for us to consider one resource's cooperation over time due to two systems' independent traffic variations. Yet, such one-resource cooperation scheme is not as widely used or efficient as the  two-resource cooperation scheme.}

To the best of our knowledge, this paper is the first attempt to study the joint energy and spectrum cooperation between different cellular systems powered by both renewable and conventional energy sources. The main contributions are summarized as follows:
\begin{itemize}
  \item \emph{Joint energy and spectrum cooperation scheme:} In Sections \ref{sec:system_model} and \ref{sec:problem_formulation}, we propose a joint energy and spectrum cooperation scheme between different cellular systems.  We provide a practical formulation of the renewable energy availability, inefficient energy and spectrum cooperation and the conventional and renewable energy costs in two systems' operational costs.

  \item  \emph{Centralized algorithm for full cooperation:} In Section \ref{sec:fullcoop}, we first consider the case of {\it full cooperation}, where  two systems belong to the same entity. We formulate the full energy and spectrum cooperation problem as a convex optimization problem, which minimizes the weighted sum cost of the two systems. We give the optimal solution to this problem in closed form. Our results show  that it is possible in this scenario, that one system shares both the spectrum and energy to the other system.

  \item \emph{Distributed algorithm for partial cooperation:} In Section \ref{sec:distributed}, we further study the case of {\it partial cooperation}, where the two systems belong to different entities and have their own interests.  We analytically characterize the partial cooperation conditions for the two systems to exchange the two resources. Under these conditions, we then propose a distributed algorithm for the two systems to gradually and simultaneously reduce their costs from the non-cooperative benchmark to the Pareto optimum. The algorithm also takes fairness into consideration, by reducing each system's cost proportionally in each iteration.

  \item \emph{Performance Evaluation:} In Section \ref{sec:numerical_examples},  we provide numerical results to validate the convergence of the distributed algorithm to the Pareto optimum and show a significant cost reduction of our proposed centralized and distributed approaches for fully and partially cooperative systems.
  \end{itemize}

In the literature,  there are some recent works studying energy cooperation in wireless systems (e.g., \cite{Chia2013,Xu2013,Gurakan2013}).  \cite{Chia2013} first considered the energy cooperation in a two-BS cellular network to minimize the total energy drawn from conventional grid subjected to certain requirements. Both off-line and on-line algorithms for the cases of unavailable and available future energy information were proposed. In \cite{Xu2013}, the authors proposed a joint communication and energy cooperation approach in coordinated multiple-point (CoMP) cellular systems powered by energy harvesting. They maximized the downlink sum-rate by jointly optimizing energy sharing and zero-forcing precoding. Nevertheless, both works \cite{Chia2013,Xu2013} only considered the cooperation within one single cellular system instead of inter-system cooperation. Another work worth mentioning is \cite{Gurakan2013}, which studies the wireless energy cooperation in different setups of wireless systems such as the one-way and two-way relay channels. However, different from our paper, which realizes the energy cooperation via wired transmission,  the energy sharing in \cite{Gurakan2013} is enabled by wireless power transfer with limited energy sharing efficiency and cooperation can only happen in one direction.

The idea of spectrum cooperation in this paper is similar to the {\it cooperative spectrum sharing} in the cognitive radio network literature (e.g.,  \cite{Simeone2008,Duan2014,Lin2011}), where secondary users (SUs) cooperate with primary users (PUs) to co-use PUs' spectrum. In order to create incentives for sharing, there are basically two approaches: {\it resource-exchange} \cite{Simeone2008,Duan2014} and {\it money-exchange} \cite{Lin2011}. For the resource-exchange approach, SUs relay traffics for PUs in exchange for dedicated spectrum resources for SUs' own communications \cite{Simeone2008,Duan2014}. Specifically, in \cite{Simeone2008}, the problem is formulated as a Stackelberg game, where the PU attempts to maximize its quality of service (QoS), while the SUs compete among themselves for transmission within the shared spectrum from the PU. In \cite{Duan2014}, the PU-SU interactions under incomplete information are modelled as a labor market using contract theory, in which the optimal contracts are designed. For the money-exchange approach, PU sells its idle spectrum to SUs. The authors in \cite{Lin2011} model the spectrum trading process as a monopoly market and accordingly design a monopolist-dominated quality price contract, where the necessary and sufficient conditions for the optimal contract are derived.

 It is worth noting that there has been another line of research on improving the energy efficiency in wireless networks by offloading traffic across different transmitters and/or systems \cite{Han2013,Yaacoub2013,Niu2010}. \cite{Han2013} studied a cognitive radio network, where the PU reduces its energy consumption by offloading part of their traffic to the secondary user (SU), while in return the PU shares its licensed spectrum bands to the SU. \cite{Yaacoub2013} and \cite{Niu2010} considered a single cellular system, in which some BSs with light traffic load can offload its traffic to the neighboring BSs and then turn off for saving energy. Although these schemes can be viewed as another approach to realize the spectrum and energy cooperation, they are different from our solution with direct  joint energy and spectrum sharing. In these works, there is no direct energy transfer between systems and the systems needs to be significantly changed in order to realize the proposed protocol.

Compared to the above existing works, the novelty of this paper is twofold. First, we provide a comprehensive study on the joint energy and spectrum cooperation by taking the uncertainty of  renewable energy  and the relationship between the two resources into account. Second, we consider the  conflict of interests between systems and propose both centralized and distributed algorithms for the cases of fully and partially cooperative systems.

\section{System Model}\label{sec:system_model}

We consider two neighbouring cellular systems that operate over different frequency bands. The two systems can either belong to the same entity (e.g., their associated operators are merged as a single party like Sprint and T-Mobile in some states of  US \cite{Sprint2013}) or relate to different entities. For the purpose of initial investigation, as shown in Fig.  \ref{sys_diagram}, we focus our study on the downlink transmission of two (partially) overlapping cells each belonging to one cellular system.\footnote{Our results can be extended to the multi-cell setting for each system by properly pairing the BSs in different systems.} In each cell $i\in\{1,2\}$, there is a single-antenna BS serving $K_i$ single-antenna mobile terminals (MTs). The sets of MTs associated with the two BSs are denoted by $\mathcal{K}_1$ and $\mathcal{K}_2$, respectively, with $|\mathcal{K}_1| = K_1$ and $|\mathcal{K}_2| = K_2$. We consider that the two BSs purchase energy from both conventional grid and their dedicated local renewable utility firms. For example, as shown in Fig. \ref{sys_diagram},  the local utility firm with solar source connects to BS 1 via a direct current (DC)/DC converter, the other one with wind source connects to BS 2 through an alternating current (AC)/DC converter, and power grid connects to both BSs by using AC/DC converters.  By combining energy from the two different supplies, BS 1 and BS 2 can operate on their respective DC buses. Here, different  power electronic circuits (AC/DC converter, DC/DC converter, etc.), which connect the BSs, renewable utility firms and the grids, are based on the types of power line connections (AC or DC buses) and the properties of different nodes (e.g., the BS often runs on a DC bus \cite{Huawei}).

\begin{figure*}[t]
  \centering
  \includegraphics[width=15cm]{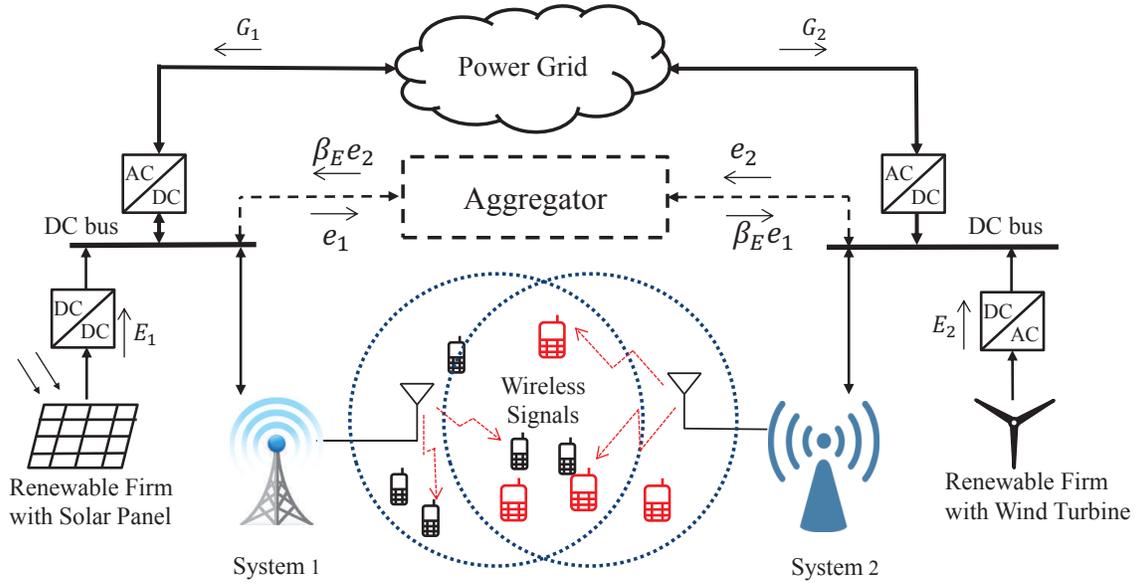}\\
  \caption{Two neighbouring cellular systems powered by power grid and renewable energy with joint energy and spectrum cooperation.}\label{sys_diagram}
\end{figure*}

We consider a time-slotted model in this paper, where the energy harvesting rate remains constant in each slot and may change from one slot to another. In practice, the harvested energy from solar and wind   remains constant over a window of seconds and we choose our time slot of the same duration.\footnote{Without loss of generality, we can further normalize the duration of each slot to be a unit of time so that the terms ``power'' and ``energy'' can be used interchangeably.} We further assume that a BS's operation in one time slot is based on its observation of the energy availability, channel conditions, traffic loads, etc., and is independent from its operation in another time slot. This is reasonable as current energy storage devices are expensive and often capacity-limited compared to power consumption of the BSs and many cellular systems do not rely on storage for dynamic energy management.\footnote{Note that the existing storage devices in the wireless systems today are generally  utilized for backup in case of the power supply outage, instead of for dynamic energy management with frequent power charing and discharging. On the other hand, the energy storage devices able to charge and discharge on the small time scales are  capacity-limited compared to power consumption of the BSs and are also expensive.  Hence, in this paper, we assume that the storage is not used at the BS, provided that the power supply from the grid is reliable.}
 Therefore, we can analyze the two systems' cooperation problem in each time slot individually. In the following, we first introduce the operation details of two systems' energy and spectrum cooperation, and then  propose the downlink transmission model for both systems.

\subsection{Energy Cooperation Model}
Recall that each BS can purchase energy from both conventional grid and renewable utility firms. The two different  types of energy supplies are characterized as follows.
\begin{itemize}
  \item {\it Conventional energy from the power grid}:  Let the energy drawn by BS $i\in\{1,2\}$ from the grid be denoted by $G_i \ge 0$. Since the practical energy demand from an individual BS is relatively small compared to the whole demand and production of the power grid network, the available energy from grid is assumed to be infinite for BS $i$. We denote $\alpha^{G}_i> 0$ as the price per unit energy purchased from grid by BS $i$. Accordingly, BS $i$'s payment to obtain energy  from grid  is $\alpha^{G}_iG_i$.
  \item {\it Renewable energy from the renewable utility firm}:  Let the energy purchased by BS $i\in\{1,2\}$ from the dedicated local renewable utility firms be denoted as $E_i> 0$.  Different from conventional energy from the grid, the local renewable energy firm is capacity-limited and subject to uncertain power supply due to environmental changes. Therefore, BS $i$ cannot purchase more than $\bar{E}_i$, which is the electricity production of the utility firm produces in the corresponding  time slot. That is,
        \begin{align}\label{renewable}
        E_i \le \bar{E}_i, i \in \{1,2\}.
        \end{align}
      Furthermore, we denote $\alpha^{E}_i > 0$ as the price of renewable energy at BS $i$ and BS $i$'s payment to obtain energy  from renewable utility firm  is $\alpha^{E}_iE_i$.
  \end{itemize}
By combing the conventional and renewable energy costs, the total cost at BS $i$ to obtain energy $G_i+E_i$ is thus denoted as
\begin{align}\label{cost}
 C_i=\alpha^{E}_iE_i+\alpha^{G}_iG_{i},~i \in \{1,2\}.
\end{align}
The price to obtain a unit of renewable energy is lower than that of conventional energy (i.e., $\alpha_i^E<\alpha_i^G, \forall i\in\{1,2\}$).   This can be valid in reality thanks to governmental subsidy, potential environmental cost of conventional energy, and the high cost of delivering conventional energy to remote areas, etc.

 Next, we consider the energy cooperation between the two systems.    Let the transferred energy from BS 1 to BS 2 be denoted as $e_1\geq0$ and that from BS 2 to BS 1 as $e_2\geq0$. Practically, the energy cooperation between two systems can be implemented by connecting  the two BSs  to a common aggregator as shown in Fig. \ref{sys_diagram}.\footnote{{ Aggregator is a virtual entity in the emerging smart gird that aggregates and controls the generation and  demands at distributed end users (e.g., BSs in cellular systems) \cite{Gkatzikis2013}. In order to manage these distributed loads more efficiently, the aggregator allows the end users to either draw or inject energy from/to it under different demand/supply conditions, by leveraging the two-way information and energy flows supported by the emerging smart grid \cite{Fang2012}. By utilizing the two-way energy transfer between the end users and the aggregator, the energy sharing between the BSs  can be enabled. With the advancement of smart grid technologies, we envision that the two-way energy transfer here would not induce additional cost.}}
 When BS $i$ wants to share energy with BS $\bar{\imath}$, where  $\bar{\imath}\in\{1,2\}\setminus\{i\}$, BS $i$ first notifies BS $\bar{\imath}$ the transmitted energy amount $e_i$. Then, at the appointed time, BS $i$ injects $e_i$ amount of energy to the aggregator and BS $\bar{\imath}$ draws $\beta_Ee_i$ amount of energy out. Thus,  energy sharing without disturbing balance in the total demand and supply can be accomplished via the aggregator. Here, $0\leq\beta_E\leq1$ is the  energy transfer efficiency factor between the two BSs that specifies the unit energy loss through the aggregator for the transferred amount of  power.\footnote{{ It is worth noting that there also exists an alternative  approach to realize the energy sharing by direct power-line connection between the BSs \cite{Chia2013}. In this approach, since dedicated power lines may need to be newly deployed, it may require higher deployment cost than the aggregator-assisted energy sharing.  Note that such an approach has been implemented in the smart grid deployments, e.g., to realize the energy transfer among different micro-grids \cite{Saad2011}.}}

\begin{figure}[t]
  \centering
    \includegraphics[width=8cm]{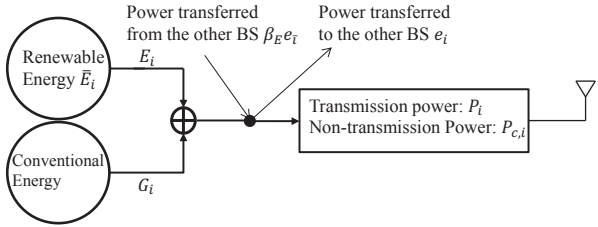}\\
  \caption{Energy management schematic of BS $i$.}\label{Energy_Model}
\end{figure}

As depicted in Fig. \ref{Energy_Model}, the energy management scheme at each BS $i\in\{1,2\}$ operates as follows.  First, at the beginning of each time slot, BS $i$ purchases the conventional energy $G_i$  and the renewable energy $E_i$. Second, it performs energy cooperation by either transferring $e_i$ amount of energy to BS $\bar{\imath}$ or collecting the exchanged energy $\beta_Ee_{\bar{\imath}}$ from BS $\bar{\imath}$. Finally, BS $i$ consumes a constant non-transmission power $P_{c,i}$ to maintain its routine operation and a transmission power $P_i$ for flexible downlink transmission. By considering transmission power, non-transmission power, and shared power between BSs, we can obtain the total power consumption at BS $i$, which is constrained by the total power supply:
\begin{align}
  \frac{1}{\eta}P_{i}+P_{c,i}\leq E_i+G_i+\beta_E e_{\bar{\imath}}-e_{i},~i \in \{1,2\},\label{energy}
\end{align}
 where $0<\eta\le 1$ is the power amplifier (PA) efficiency. Since $\eta$ is a constant, we normalize it as $\eta=1$ in the sequel.

\subsection{Spectrum Cooperation Model}
We now explain the spectrum cooperation by considering two cases: adjacent and non-adjacent frequency bands. First, consider the case of  adjacent frequency bands in Fig. \ref{bandwidth_sharing}. As shown in Fig. \ref{bandwidth_sharing}, BS 1 and BS 2 operate in the blue and red shaded frequency bands  $W_1$ and $W_2$, respectively. Between $W_1$ and $W_2$ a guard band $W_G$ is inserted to avoid interference due to out-of-band emissions. Let the shared spectrum bandwidth from BS 1 to BS 2 be denoted as $w_1\geq0$ and that from BS 2 to BS 1 as $w_2\geq0$. In this case,  the shared bandwidth from BS $i$ (i.e., $w_i$) can be fully used at BS $\bar{\imath}$, $i \in \{1,2\}$. This can be implemented by carefully moving the guard band as shown in Fig. \ref{bandwidth_sharing}.   When BS 2 shares a bandwidth  $w_2$ to BS 1, the green shaded guard band with bandwidth $W_G$ is moved accordingly between $W_1+w_2$ and $W_2 - w_2$,  such that the shared spectrum is fully utilized.

\begin{figure}
        \centering
        \begin{subfigure}[b]{0.5\textwidth}
  \centering
  \includegraphics[width=8cm]{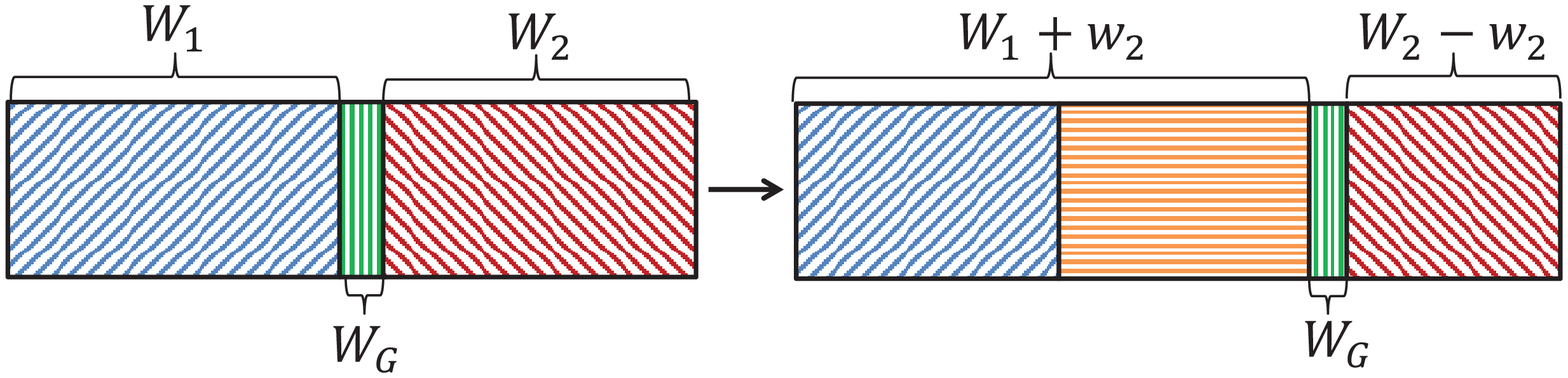}\\
  \caption{Adjacent frequency bands.}\label{bandwidth_sharing}
        \end{subfigure}\\
        \begin{subfigure}[b]{0.5\textwidth}
  \centering
  \includegraphics[width=8cm]{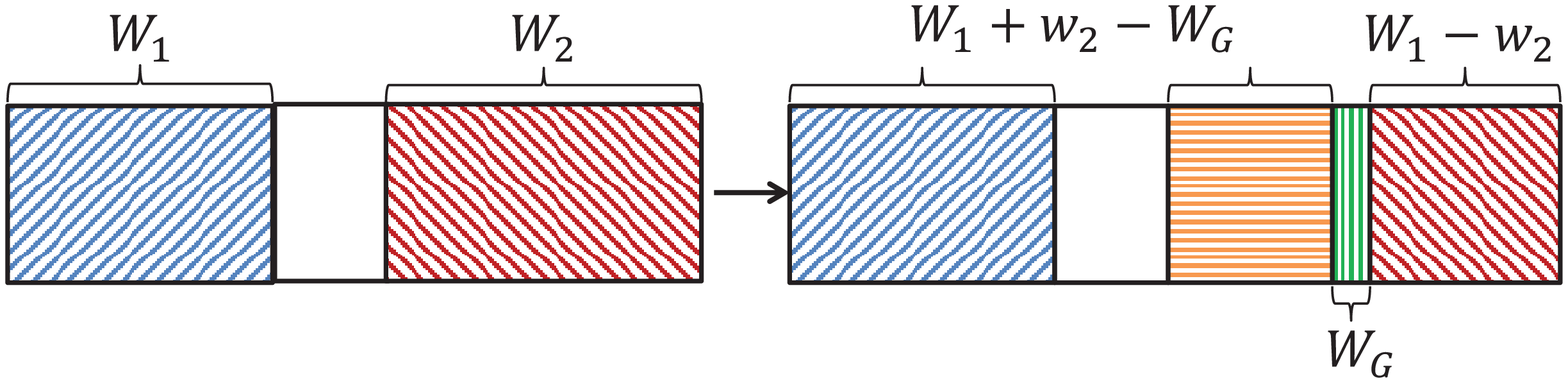}\\
  \caption{Non-adjacent frequency bands.}\label{bandwidth_sharing2}
        \end{subfigure}
        \caption{An example of spectrum cooperation between two BSs.}\label{fig:bandwidthcooperation}

\end{figure}
Next, consider the case of non-adjacent frequency bands in Fig. \ref{bandwidth_sharing2}.  For this non-adjacent frequency band, a guard band is also needed to avoid the inter-system interference. Fig. \ref{bandwidth_sharing2} shows an example of the spectrum cooperation when BS 2 shares a bandwidth $w_2$ to BS 1. After spectrum cooperation, the total usable spectrum of BS 1 is $W_1+w_2-W_G$, since a green shaded guard band $W_G$ is inserted between $W_1 +w_2-W_G$ and $W_2 -w_2$ to avoid the inter-system interference. As a result, spectrum cooperation loss will occur.\footnote{ It is technically challenging  to gather non-adjacent pieces of bandwidth together at one BS. To overcome this issue, the carrier aggregation solution for the LTE-Advanced system \cite{Qualcomm} can be utilized here. }

For the ease of investigation, in this paper we only focus on the former case of adjacent frequency bands.\footnote{It should be noticed that our result can also be extended to the non-adjacent bandwidth case by considering the bandwidth loss due to the guard band.} We define a spectrum cooperation factor $\beta_B \in \{0,1\}$, for which $\beta_B=1$ denotes that spectrum cooperation is implementable  between BSs and $\beta_B=0$ represents that spectrum cooperation is infeasible.   Considering the spectrum cooperation between the BSs, the bandwidth used by BS $i$ can be expressed as
\begin{align}
  B_i\leq W_i+\beta_Bw_{\bar{\imath}}-w_i,~i\in\{1,2\}.\label{bandwidthconsumption}
\end{align}
Note that in our investigated spectrum cooperation, the (shared) spectrum resources can only be utilized by either BS 1 or BS 2 to avoid the interference between the two systems. If the same spectrum resources can be used by the two systems at the same time, then the spectrum utilization efficiency can be further improved while also  introducing inter-system interference \cite{Gesbert2010}. In this case, more sophisticated interference coordination should be implemented, which is beyond the scope of this work.

\subsection{Downlink Transmission Under Energy and Spectrum Cooperation}
We now introduce the downlink transmission at each BS by incorporating the energy and spectrum cooperation. We consider a flat fading channel model for each user's downlink transmission, and denote the channel gain from BS $i\in\{1,2\}$ to its associated MT $k$ as $g_k$, $k\in\mathcal{K}_i$, which in general includes the pathloss, shadowing and antenna gains. Within each system, we assume orthogonal transmission to support multiple MTs, e.g., by applying orthogonal frequency-division multiple access (OFDMA). Accordingly, the signal-to-noise-ratio (SNR) at each MT $k$ is given by
\begin{align}\label{eqn:6}
\mathtt{SNR}_{k} = \frac{g_k p_k}{b_k N_0},  k\in\mathcal{K}_1\cup\mathcal{K}_2,
\end{align}
where $N_0$ denotes the power spectral density (PSD) of the additive white Gaussian noise (AWGN), $p_k\ge 0$ and $b_k\ge 0$ denote the allocated power and bandwidth to MT $k\in\mathcal{K}_1\cup\mathcal{K}_2$, respectively. We can aggregate all the transmission power and bandwidth used by each BS as (cf. (\ref{energy}) and (\ref{bandwidthconsumption}))
\begin{align}
P_i = \sum_{k\in\mathcal{K}_i}p_k,~B_i =\sum_{k\in\mathcal{K}_i}b_k,~i \in \{1,2\}\nonumber.
\end{align}
Note that the bandwidth and power allocation are performed in each slot on the order of seconds, which is much longer than the coherence time of wireless channels (on the order of milliseconds). As a result, the SNR defined in (\ref{eqn:6}) is time-averaged over the dynamics of wireless channels, and thus the fast fading is averaged out from the channel gain $g_k$'s.

To characterize the QoS requirements of each MT, we define its performance metric as a utility function $u_{k}(b_k,p_k)$, and assume that it satisfies the following three properties:
\begin{itemize}
  \item[1)] The utility function is non-negative, i.e., $u_{k}(b_k,p_k) \ge 0,~\forall p_k\geq 0$, $b_k\geq0$, where $u_{k}(0,p_k) = 0$ and $u_{k}(b_k,0) = 0$;
  \item[2)] The utility increases as a function of allocated power and bandwidth, i.e., $u_{k}(b_k,p_k)$ is monotonically increasing with respect to $b_k$ and $p_k$, $\forall p_k\geq 0$, $b_k\geq0$;
  \item[3)] The marginal utility decreases as the allocated power and bandwidth increase, i.e., $u_{k}(b_k,p_k)$ is jointly concave over $b_k$ and $p_k$, $\forall p_k\geq 0$, $b_k\geq0$.
\end{itemize}
For example, the achievable data rate at MT $k$ defined as follows is a feasible utility function that satisfies the above three properties \cite{Tse2005}
\begin{align}
u_{k}(b_k,p_k)=b_k\log_2(1+\mathtt{SNR}_{k})=b_k\log_2\bigg(1+\frac{g_kp_k}{b_kN_0}\bigg)\label{capacity}.
\end{align}
In the rest of this paper, we employ the utility function in (\ref{capacity}) for all MTs in the two systems and averting to another function will not change our main engineering insights.
Due to the fact that most cellular network services are QoS guaranteed (e.g., minimum date rate in video call), we ensure the QoS requirement at each MT $i$ by setting a minimum utility threshold {$r_k > 0$, $k\in\mathcal{K}_1\cup\mathcal{K}_2$}. The value of $r_k$ is chosen according to the type of service at MT $k$. Accordingly, the resultant QoS constraint is given by
\begin{align}
b_k\log_2\left(1+\frac{g_k p_k}{b_k N_0}\right) \ge r_k,~\forall k\in\mathcal{K}_1\cup\mathcal{K}_2.\label{QoS}
\end{align}

\begin{table}[t]\caption{List of notations and their physical meanings}
\centering
\begin{tabular}{|l | p{7.5cm}| }
\hline
$\bar{E}_i$ & Maximum purchasable  renewable energy at BS $i$\\
$W_i$  & Available bandwidth at BS $i$\\
$P_{c,i}$  & Constant non-transmission power consumption at BS $i$\\
$\alpha_i^E$  & Price of per-unit renewable energy for BS $i$\\
$\alpha_i^G$ & Price of per-unit conventional energy for BS $i$\\
$\beta_E$ & Energy cooperation efficiency between two BSs\\
$\beta_B$ & Spectrum cooperation factor between two BSs\\
$r_k$ & QoS requirement of MT $k$\\
$E_i$  &Renewable energy drawn at BS $i$\\
$G_i$ & Conventional energy drawn from the grid by BS $i$\\
$p_k$  & Allocated power to MT $k$\\
$b_k$  & Allocated bandwidth to MT $k$\\
$e_i$  & Shared energy from BS $i$  to BS $\bar{\imath}$\\
$w_i$  & Shared spectrum from BS $i$  to BS $\bar{\imath}$\\
$\mv{x}^{\mathrm{in}}_i$  & Intra-network energy and bandwidth
allocation vector at BS $i$ consisting of $G_i$, $E_i$, $p_k$ and $b_k, \forall k\in\mathcal{K}_i$\\
$\mv{x}^{\mathrm{ex}}$ & Inter-network energy and spectrum cooperation vector consisting of $e_1, e_2, w_1$ and $w_2$\\
$\mv{x}$ & An aggregated vector consisting all the decision variables of the two BSs\\
\hline
\end{tabular}
\label{tab:TableOfNotationForMyResearch}

\end{table}

\section{Problem Formulation}\label{sec:problem_formulation}
 We aim to reduce the costs $C_1$ and $C_2$ in (\ref{cost}) at both BSs while guaranteeing the QoS requirements of all MTs. We denote the intra-system decision vector for BS $i\in\{1,2\}$ as $\mv{x}_i^{\rm{in}}$, which consists of its energy drawn from renewable energy $E_i$, energy drawn from the grid $G_i$, power allocation $p_k$'s, and bandwidth allocation $b_k$'s with $k\in\mathcal{K}_i$. We also denote $\mv{x}^{\rm{ex}}=[e_1,e_2,w_1,w_2]^T$ as the inter-system energy and spectrum cooperation vector. For convenience, we aggregate all the decision variables of the two BSs as $\mv{x}\triangleq[{\mv{x}_1^{\rm{in}}}^T, {\mv{x}_2^{\rm{in}}}^T, {{\mv{x}^{\rm{ex}}}}^{T}]^T$. All the notations used in this paper are summarized and explained in Table \ref{tab:TableOfNotationForMyResearch} for the ease of reading.

It can be shown that the two systems (if not belonging to the same entity)  have conflicts in cost reduction under the joint energy and spectrum cooperation. For example, if BS 1 shares both energy and spectrum to BS 2, then the cost of BS 2 is reduced while the cost of BS 1 increases.  To characterize such conflicts, we  define the achievable cost region under the joint energy and spectrum cooperation as the cost tuples that the two BSs can achieve simultaneously, which is explicitly characterized by
\begin{align}\label{region}
  \mathcal{C}\triangleq \bigcup_{\mv{x}\ge \bf{0},\mv{x}\in\mathcal{X}} \{(c_1,c_2): C_i(\mv{x})\leq c_i,i \in \{1,2\}\},
\end{align}
where $\mathcal{X}$ is the feasible set of $\mv{x}$  specified by (\ref{renewable}), (\ref{energy}), (\ref{bandwidthconsumption}) and (\ref{QoS}), and $C_i(\mv{x})$ is the achieved cost of BS $i \in \{1,2\}$ in (\ref{cost}) under given $\mv{x}$. The boundary of this region is then called the {\it Pareto boundary}, which consists of the Pareto optimal cost tuples at which it is impossible to decrease one's cost without increasing the other's. Since the feasible region $\mathcal{X}$ can be shown to be convex and the cost in (\ref{cost}) is affine, the cost region in (\ref{region}) is convex. Also, because the Pareto optimal points of any convex region can be found by solving a series of weighted sum minimization problems with different weights \cite{Boyd04}, we can achieve different Pareto optimal cost tuples by solving the following weighted sum cost minimization problems

\begin{subequations}
\begin{align}
\mathrm{(P1)}:~\mathop{\mathtt{min.}}_{\mv{x}\ge \bf{0}}&~~
\sum_{i=1}^{2}\gamma_i(\alpha^{E}_iE_i+\alpha^{G}_iG_{i})\label{P1:obj}\\
\mathtt{s.t.}
&~ \sum_{k\in\mathcal{K}_i}p_{k}+P_{c,i}\leq E_i+G_i+\beta_{E} e_{\bar{\imath}}-e_{i},~ \forall i \in \{1,2\},\label{P1:cons1}\\
&~\sum_{k\in\mathcal{K}_i}b_k\leq W_i+\beta_Bw_{\bar{\imath}}-w_i, \forall i \in \{1,2\},\label{P1:cons3}\\
&~E_{i}\leq \bar{E}_i,~\forall i \in \{1,2\},\label{P1:cons2}\\&~~ b_{k}\log_2\bigg(1+\frac{g_{k}p_{k}}{ b_{k}N_0}\bigg)\geq r_{k},~ \forall k\in \mathcal{K}_1\cup \mathcal{K}_2\label{P1:cons4},
\end{align}
\end{subequations}
where $\gamma_i\ge 0,i \in \{1,2\}$ is the cost weight for BS $i$, which specifies the trade-offs between the two BSs' costs. By solving (P1) given different $\gamma_i$'s, we can characterize the entire Pareto boundary of the cost region. For the solution of the optimization problem, standard convex optimization techniques such as the interior point method can be employed to solve (P1) \cite{Boyd04}. However, in order to gain more engineering insights, we propose an efficient algorithm for problem (P1) by applying the Lagrange duality method in Section \ref{sec:fullcoop}. Before we present the solution for (P1), in this section, we first consider a special case where there is neither energy nor spectrum cooperation between the two systems (i.e., $\beta_E=\beta_B=0$). This serves as a performance benchmark for comparison with  fully or partially cooperative systems  in Sections \ref{sec:fullcoop} and \ref{sec:distributed}.

\subsection{Benchmark Case: Non-cooperative Systems}\label{sec:nocoop}
With $\beta_E=\beta_B=0$, the two systems will not cooperate and the optimal solution of (P1) given any $\gamma_i$'s is attained with zero inter-system exchange, i.e.,  $w_1 = w_2 = e_1 = e_2 = 0$. In this case, the constraints in (\ref{P1:cons1}) and (\ref{P1:cons3}) reduce to $\sum_{k\in\mathcal{K}_i}p_{k}+P_{c,i}\leq E_i+G_i$ and $\sum_{k\in\mathcal{K}_i}b_{k}\leq W_i,~\forall i \in \{1,2\}$, respectively. It thus follows that the intra-system energy and bandwidth allocation vectors $\mv{x}_1^{\rm{in}}$ and $\mv{x}_2^{\rm{in}}$ are decoupled in both the objective and the constraints of problem (P1). As a result, (P1) degenerates to two cost-minimization problems as follows (one for each BS $i\in\{1,2\}$):
\begin{subequations}
\begin{align}
\mathrm{(P2)}:~\mathop{\mathtt{min.}}_{\mv{x}^{\rm{in}}_{i}\geq0} &~
 ~ \alpha^{E}_iE_i+\alpha^{G}_iG_{i}\nonumber\\
\mathtt{s.t.}&~~ \sum_{k\in\mathcal{K}_i}p_{k}+P_{c,i}\leq E_i+G_i\label{p1cons1:nocoop},\\
&~~E_{i}\leq \bar{E}_i\label{p1cons2:nocoop},\\
&~~\sum_{k\in\mathcal{K}_i}b_{k}\leq W_i\label{p1cons3:nocoop},\\
&~~b_{k}\log_2\bigg(1+\frac{g_{k}p_{k}}{ b_{k}N_0}\bigg)\geq r_{k},~\forall  k\in \mathcal{K}_i.\label{p1cons4:nocoop}
\end{align}
\end{subequations}
{  Note that Problem (P2) is always feasible due to the fact that the BS can purchase energy from the grid without limit. Therefore, we can always find one feasible solution to satisfy all the constraints in (\ref{p1cons1:nocoop})-(\ref{p1cons4:nocoop}). } It is easy to show that at the optimality of problem (P2), the constraints  (\ref{p1cons3:nocoop}) and (\ref{p1cons4:nocoop}) are both tight, otherwise, one can reduce the cost by reducing the allocated power $p_k$ (and/or increasing the allocated bandwidth $b_k$) to MT $k$. Then, the power allocation for each user can be expressed as
\begin{align}
p_k=\frac{b_kN_0}{g_k}\bigg(2^{\frac{r_k}{b_k}}-1\bigg),~\forall k\in\mathcal{K}_i.\label{energy:nocoop}
\end{align}
By substituting (\ref{energy:nocoop}) into (P2) and applying the Karush-Kuhn-Tucker (KKT) condition, we have the closed-form optimal solution to (P2) in the following proposition. { Note that the optimal solution is unique, since the constraints in (\ref{p1cons4:nocoop}) are strictly convex over $b_k$'s and $p_k$'s, $\forall k\in\mathcal{K}$.}
\begin{proposition}\label{prop_nocoop}
The optimal bandwidth allocation for (P2) is given by
\begin{align}
&b_k^\star=\frac{\ln2\cdot r_k}{\mathcal{W}(\frac{1}{e}(\frac{\nu_i^\star
g_k}{N_0}-1))+1}\label{Lambertx:nocoop}
  ,~\forall k\in\mathcal{K}_i,
  \end{align}
where $\mathcal{W}(\cdot)$ is Lambert $\mathcal{W}$ function \cite{Corless1996} and $\nu_i^\star\ge 0$ denotes the water level that satisfies $\sum_{k\in\mathcal{K}_i}b_k^{\star}= W_i$. Furthermore, the optimal power allocation and energy management in (P2) are given by
  \begin{align}
&p_k^\star=\frac{b_k^\star N_0}{g_k}\bigg(2^{\frac{r_k}{b_k^\star}}-1\bigg),~\forall k\in\mathcal{K}_i,\nonumber\\
&E_i^\star=\max\bigg(\sum_{k\in\mathcal{K}_i}p^\star_k+P_{c,i},\bar{E}_i\bigg),\nonumber\\
&G_i^{\star}=\max\left(\sum_{k\in\mathcal{K}_i}p^\star_k+P_{c,i}-\bar{E}_i,0\right).\nonumber
\end{align}
\end{proposition}
\begin{proof}
See Appendix \ref{no_coop_prop}.
\end{proof}

In Proposition \ref{prop_nocoop}, the bandwidth allocation $b_k^{\star}$ can be interpreted as waterfilling over different MTs with $\nu_i^{\star}$ being the water level, and the power allocation $p_k^{\star}$ follows from (\ref{energy:nocoop}). Furthermore, the optimal solution of $E_i^{\star}$ and $G_i^{\star}$ indicate that BS $i$ first purchases energy from the renewable energy firm, and (if not enough) then  from the grid. This is intuitive due to the fact that the renewable energy is cheaper  ($\alpha_i^E<\alpha_i^G$).

\section{Centralized Energy and Spectrum Cooperation for Fully Cooperative Systems}\label{sec:fullcoop}

In this section, we consider problem (P1) with given weights $\gamma_1$ and $\gamma_2$ for the general case of $\beta_B\in\{0,1\}$ and $0\leq\beta_E\leq1$. This corresponds to the scenario where the two BSs belong to the same entity and thus can fully cooperate to solve (P1) to minimize the weighted sum cost. Similar to (\ref{energy:nocoop}) in (P2), we can show that the QoS constraints in (\ref{P1:cons4}) should always be tight for the optimal solution of (P1). As a result, the power allocation for each MT in (P1) can also be expressed as (\ref{energy:nocoop}) for all $i \in \{1,2\}$. By substituting (\ref{energy:nocoop}) into the power constraint (\ref{P1:cons1}) in (P1) and then applying the Lagrange duality method, we obtain the closed-form solution to (P1) in the following proposition.

\begin{proposition}\label{optimalP1}
The optimal bandwidth and power allocation  solutions to problem (P1) are given by
\begin{align}
  b_{k}^{\star}=&\frac{\ln 2\cdot r_k}{\mathcal{W}\left(\frac{1}{e}\left(\frac{\lambda_i^{\star} g_k}{\mu_i^{\star}N_0}-1\right)\right)+1},~\forall k\in\mathcal{K}_1\cup\mathcal{K}_2,
  \nonumber\\
  p_k^\star=&\frac{b_k^\star N_0}{g_k}\left(2^{\frac{r_k}{b_k^\star}}-1\right),~\forall k\in\mathcal{K}_1\cup\mathcal{K}_2,\nonumber
\end{align}
where $\lambda_i^{\star}$ and $\mu_i^{\star}$ are non-negative constants (dual variables) corresponding to the power constraint (\ref{P1:cons1}) and the bandwidth constraint (\ref{P1:cons3}) for BS $i\in\{1,2\}$, respectively.\footnote{The optimal dual variables $\{\lambda_i^{\star}\}_{i=1}^2$ and $\{\mu_i^{\star}\}_{i=1}^2$ can be obtained by solving the dual problem of (P1) as explained in Appendix \ref{proofoptimalP1}.} Moreover,  the optimal spectrum sharing between the two BSs are
\begin{align}
 w_i^\star= \max{\bigg(-\sum_{k\in\mathcal{K}_i}b_k^\star+W_i,0\bigg)},\forall i \in \{1,2\}.\label{optimal_spectrum}
\end{align}
Finally, the optimal energy decisions at two BSs $\{E_i^\star\},\{G_i^\star\}$ and  $\{e_i^\star\}$ are the solutions to the following problem.
\begin{align}
\mathrm{(P3)}:\nonumber\\
\mathop{\mathtt{min.}}_{\{E_i,G_i,e_i\}} & \sum_{i=1}^{2}\gamma_i(\alpha^{E}_iE_i+\alpha^{G}_iG_{i})\nonumber\\
\mathtt{s.t.}
& \sum_{k\in\mathcal{K}_i}p_{k}^{\star}+P_{c,i}= E_i+G_i+\beta_E e_{\bar{\imath}}-e_{i}, ~ \forall i \in \{1,2\}\nonumber,\\
&0\leq E_i\leq \bar{E}_i, ~G_i\ge0,~e_i\ge0,~ \forall i \in \{1,2\}.  \nonumber
\end{align}
\end{proposition}
\begin{proof}
See Appendix \ref{proofoptimalP1}.
\end{proof}

Note that problem (P3) is a simple linear program (LP) and thus can be solved by existing software such as CVX  \cite{Grant2011}. Also note that there always exists an optimal solution of $\{e_i^\star\}$ in (P3) with $e_1^\star\cdot e_2^\star = 0$.{\footnote{{ If   $e_1^{\star} \cdot e_2^{\star} = 0$ does not hold,  we can find another feasible  energy cooperation solution $e_i^{{\star}'}=e_i^\star-\min(e_1^\star,e_2^\star),~\forall i\in\{1,2\}$, with $e_1^{\star' }\cdot e_2^{\star' }= 0$, to achieve no larger weighted sum cost.  }}}
It should be noted that the optimal solution in Proposition \ref{optimalP1} can only be obtained in a centralized manner. { Specifically, to perform the joint energy and spectrum cooperation, the information at both systems (i.e.,  the energy price $\alpha^E_i$ and $\alpha_i^G$, the available renewable energy $\bar E_i$, the circuit power consumption $P_{c,i}$, the channel gain $g_k$ and their QoS requirement $\bar r_k,~\forall k\in \mathcal K_i, i\in\{1,2\}$) should be gathered at a central unit, which can be one of the two BSs or a third-party controller. Since the limited information is exchanged over the time scale of power and bandwidth allocation, which is on the order of second, while the communication block usually has a length of several milliseconds, the information exchange can be efficiently implemented.
}

It is interesting to make a comparison between the optimal solution of problem (P1) in Proposition \ref{optimalP1} and that of problem (P2) in Proposition \ref{prop_nocoop}. First, it follows from the solution of $w_i^{\star}$ in (\ref{optimal_spectrum}) that if $\beta_B=1$, then the bandwidth can be allocated in the two systems more flexibly, and thus resulting in a spectrum cooperation gain  in terms of cost reduction as compared to the non-cooperative benchmark in Proposition \ref{prop_nocoop}. Next, from the LP in (P2) with $0<\beta_E\le 1$, it is evident that the BSs will purchase energy by comparing the weighted energy prices given as $\gamma_i\alpha^{E}_i$ and $\gamma_i\alpha^{G}_i, i\in\{1,2\}$. For instance, when system $i$'s weighted renewable energy price $\gamma_i\alpha_i^E$ is higher than $\gamma_{\bar{\imath}}\alpha_{\bar{\imath}}^E$ of the other system, then this system $i$ will try to first request the other system's renewable energy rather than drawing energy from its own dedicated renewable utility firm. In contrast, for the non-cooperative benchmark in Proposition \ref{prop_nocoop}, each BS always draws energy first from its own renewable energy, and then from the grid. Therefore, the energy cooperation changes the energy management behavior at each BS, and thus results in an energy cooperation gain in terms of cost reduction. It is worth noting that to minimize the weighted sum cost in the full cooperative system, it is possible for one system to contribute both spectrum and energy resources to the other (i.e., $w_i^{\star}> 0$ and $e_i^{\star} > 0$  for any  $ i\in\{1,2\}$), or one system exchanges its energy while the other shares its spectrum in return (i.e., $w_i^{\star} > 0$ and $e_{\imath}^{\star}>  0$ for any $ i\in\{1,2\}$). These two scenarios are referred to as {\it uni-directional cooperation} and {\it bi-directional cooperation}, respectively.

\section{Distributed Energy and Spectrum Cooperation for Partially Cooperative Systems}\label{sec:distributed}
In the previous section, we have proposed an optimal centralized algorithm to achieve the whole Pareto boundary of the cost region. However, this requires the fully cooperative nature and does not apply to the scenario where the two systems have their own interests (e.g., belonging to different selfish entities). Regarding this, we proceed to present a partially cooperative system that
implements the joint energy and spectrum cooperation ($0 < \beta_E \le 1$, $\beta_B = 1$) to achieve a Pareto optimum with limited information exchange in coordination. Different from the fully cooperative system that can perform both uni-directional and bi-directional cooperation, the partially cooperative systems seek mutual benefits to decrease both systems' cost simultaneously, in which only bi-directional cooperation is feasible.\footnote{{ Due to the mutual benefit, we believe that both systems have incentives for partial cooperation. Moreover, such incentives can be further strengthened in the future wireless systems envisioned to have more expensive energy and spectrum.}}  In the following, we first analytically characterize the conditions for partial cooperation. Then, we propose a distributed algorithm that can achieve the Pareto optimality.

\subsection{Conditions for  Partial Cooperation}

We define a function $\bar{C}_i(\mv{x}^{\mathrm{ex}})$ to represent the minimum cost at BS $i$ under any given energy and spectrum cooperation scheme $\mv{x}^{\mathrm{ex}}$, which is given as:
\begin{align}
\bar{C}_i(\mv{x}^{\mathrm{ex}})=\mathop{\mathtt{min.}}_{\mv{x}_i^{\rm{in}}\geq\bm{0}} &~ \alpha^{E}_iE_i+\alpha^{G}_iG_{i}\nonumber\\
\mathtt{s.t.} &~(\rm{\ref{P1:cons1}}), (\mathrm{\ref{P1:cons3}}), (\mathrm{\ref{P1:cons2}})~{\rm{and}}~(\mathrm{\ref{P1:cons4}}).\label{P4}
\end{align}
Note that based on Proposition 4.1, we only need to consider $\mv x^{\mathrm{ex}}$ with $e_1\cdot e_2 = 0$ and $w_1 \cdot w_2=0$ without loss of optimality.\footnote{{ For any given energy and spectrum cooperation scheme $\mv x^{\mathrm{ex}}$ with $e_1\cdot e_2 \neq 0$ or $w_1 \cdot w_2 \neq 0$, we can always trivially find an alternative scheme ${\mv x^{\mathrm{ex}}}' = [e_1' ,e_2' ,w_1' ,w_2']^T$ with $e_i' = e_i -\min(e_1,e_2)$ and $w_i' = w_i -\min(w_1,w_2)$ to achieve the same or smaller cost at both systems as compared to $\mv x^{\mathrm{ex}}$, i.e., $\bar C_i(\mv x^{\mathrm{ex}'}) \le C_i(\mv x^{\mathrm{ex}}), i=1,2$. Since $e_1'\cdot e_2'=0$ and $w_1'\cdot w_2' = 0$ always hold, it suffices to only consider $\mv x^{\mathrm{ex}}$ with $e_1\cdot e_2=0$ and $w_1\cdot w_2= 0$.}} The problem in (\ref{P4}) has a similar structure as problem (P2), which is a special case of (\ref{P4}) with $\mv{x}^{\mathrm{ex}}=\mv{0}$. Thus, we can obtain its optimal solution similarly as in Proposition \ref{prop_nocoop} and the details are omitted here. { We denote the optimal solution to problem (\ref{P4}) by $E_i^{(\mv{x}^{\mathrm{ex}})}, G_i^{(\mv{x}^{\mathrm{ex}})}, \{b_k^{(\mv{x}^{\mathrm{ex}})}\}$, and $\{p_k^{(\mv{x}^{\mathrm{ex}})}\}$ and the bandwidth water-level $\nu_i^{(x^{\mathrm{ex}})}$. Furthermore, let  the optimal dual solution associated with (\ref{P1:cons1}) and (\ref{P1:cons2}) be denoted  by $\lambda_i^{(\mv{x}^{\mathrm{ex}})}$ and $\mu_i^{(\mv{x}^{\mathrm{ex}})}$, respectively.}
Then, it follows that\footnote{As will be shown later, $\mu_i^{(\mv{x}^{\mathrm{ex}})}$ and $\lambda_i^{(\mv{x}^{\mathrm{ex}})}$ can be interpreted as the marginal costs with respect to the shared energy and bandwidth between two BSs, respectively. Therefore, the result in (\ref{P4dual:nocoop}) is intuitive, since the marginal cost should be the energy price of $\alpha_i^E$ if the renewable energy is excessive to support the energy consumption and energy exchange, while the marginal cost should be $\alpha_i^G$ if the renewable energy is insufficient.
}
 \begin{align}
& \mu_i^{(\mv{x}^{\mathrm{ex}})}=
\begin{cases}
\alpha_i^E, & \sum_{k\in\mathcal{K}_i}p^{(\mv{x}^{\mathrm{ex}})}_k+P_{c,i}-\beta_E e_{\bar{\imath}}+e_{i}\le\bar{E}_i  \\
\alpha_i^G, &\sum_{k\in\mathcal{K}_i}p^{(\mv{x}^{\mathrm{ex}})}_k+P_{c,i}-\beta_E e_{\bar{\imath}}+e_{i}>\bar{E}_i
\end{cases}.\label{P4dual:nocoop}\\
&{
\lambda_i^{(\mv{x}^{\mathrm{ex}})}=\nu_i^{(\bm{x}^{\mathrm{ex}})}\cdot \mu_i^{(\bm{x}^{\mathrm{ex}})}}\label{waterlevel}
\end{align}

It is easy to verify that $\bar{C}_i(\mv{x}^{\mathrm{ex}}),i\in\{1,2\}$, is a convex function of $\mv{x}^{\mathrm{ex}}$. Therefore,  under any given $\mv{x}^{\mathrm{ex}}$, two BSs can reduce their individual cost simultaneously if and only if there exists ${\mv{x}^{\mathrm{ex}}}' = \mv{x}^{\mathrm{ex}}+\Delta \mv{x}^{\mathrm{ex}}\neq \mv{x}^{\mathrm{ex}}$ with $\Delta\mv{x}^{\mathrm{ex}}=[\Delta e_1, \Delta e_2, \Delta w_1, \Delta w_2]^T$ sufficiently small and ${\mv{x}^{\mathrm{ex}}}' \geq \mv{0}$ and ${\mv{x}^{\mathrm{ex}}}' \neq \mv{0}$ such that $\bar{C}_i({\mv{x}^{\mathrm{ex}}}') <\bar{C}_i(\mv{x}^{\mathrm{ex}}),\forall i\in\{1,2\}$. In particular, by considering the non-cooperative benchmark system with $\mv{x}^{\mathrm{ex}}=\mv{0}$, it is inferred that partial cooperation is feasible if and only if there exists ${\mv{x}^{\mathrm{ex}}}' \geq \mv{0}$ and ${\mv{x}^{\mathrm{ex}}}' \neq \mv{0}$  such that $\bar{C}_i({\mv{x}^{\mathrm{ex}}}') < \bar{C}_i(\mv{0})$. Based on these observations, we are ready to investigate the conditions for partial cooperation by checking the existence of such ${\mv{x}^{\mathrm{ex}}}'$.
First, we derive BS $i$'s cost change  $\bar{C}_i({\mv{x}^{\mathrm{ex}}}') - \bar{C}_i(\mv{x}^{\mathrm{ex}})$ analytically when the energy and spectrum cooperation decision changes from any given $\mv{x}^{\mathrm{ex}}$ to ${\mv{x}^{\mathrm{ex}}}' = \mv{x}^{\mathrm{ex}}+\Delta \mv{x}^{\mathrm{ex}}$ with sufficiently small $\Delta\mv{x}^{\mathrm{ex}}$. We have the following proposition.
\begin{lemma}\label{proposition:5.1}
Under any given $\mv{x}^{\mathrm{ex}}$, BS $i$'s cost change by adjusting the energy and spectrum cooperation decisions is expressed as
\begin{align}
  \bar{C}_i(\mv{x}^{\mathrm{ex}}+\Delta \mv{x}^{\mathrm{ex}})-\bar{C}_{i}(\mv{x}^{\mathrm{ex}})=\nabla \bar{C}_i(\mv{x}^{\mathrm{ex}})^T\Delta\mv{x}^{\mathrm{ex}},\label{1storder}
\end{align}
where $\Delta\mv{x}^{\mathrm{ex}}$ is sufficiently small, $\mv{x}^{\mathrm{ex}}+\Delta\mv{x}^{\mathrm{ex}} \ge \mv{0}$, and
\begin{align}\label{partialderivatives}
\nabla \bar{C}_i(\mv{x}^{\mathrm{ex}})=\bigg[\frac{\partial \bar{C}_i(\mv{x}^{\mathrm{ex}})}{\partial e_1}, \frac{\partial \bar{C}_i(\mv{x}^{\mathrm{ex}})}{\partial e_2},\frac{\partial \bar{C}_i(\mv{x}^{\mathrm{ex}})}{\partial w_1}, \frac{\partial \bar{C}_i(\mv{x}^{\mathrm{ex}})}{\partial w_2}\bigg]^T.
\end{align}
Here, $\frac{\partial \bar{C}_i(\bm{x}^{\mathrm{ex}})}{\partial e_i} =\mu_i^{(\bm{x}^{\mathrm{ex}})},~
                             \frac{\partial \bar{C}_i(\bm{x}^{\mathrm{ex}})}{\partial e_{\bar{\imath}}} =-\beta_E\mu_i^{(\bm{x}^{\mathrm{ex}})}$, $\frac{\partial \bar{C}_i(\bm{x}^{\mathrm{ex}})}{\partial w_i}=\lambda_i^{(\bm{x}^{\mathrm{ex}})}$ and $\frac{\partial \bar{C}_i(\bm{x}^{\mathrm{ex}})}{\partial w_{\bar{\imath}}}=-\lambda_i^{(\bm{x}^{\mathrm{ex}})}$ can be interpreted as the marginal costs at BS $i$ with respect to the energy and spectrum cooperation decisions $e_i, e_{\bar{\imath}}, w_i$ and $w_{\bar{\imath}}$, respectively.
\end{lemma}
\begin{proof}
See Appendix \ref{appendix:proof_lemma_5.1}.
\end{proof}

Next, based on Lemma \ref{proposition:5.1}, we obtain the conditions for which the two BSs' costs can be decreased at the same time under any given $\mv{x}^{\rm{ex}}$, by examining whether there exists sufficiently small $\Delta \mv{x}^{\mathrm{ex}}\neq \mv{0}$ with $\mv{x}^{\mathrm{ex}}+\Delta\mv{x}^{\mathrm{ex}} \ge \mv{0}$ such that $\nabla \bar{C}_i(\mv{x}^{\mathrm{ex}})^T\Delta\mv{x}^{\mathrm{ex}} < 0$ for both $i=1,2$.

 \begin{proposition}\label{theorem1}
For any given $\mv{x}^{\mathrm{ex}}$, the necessary and sufficient conditions that the two BSs' costs can be decreased at the same time are given as follows:
\begin{itemize}\itemsep0.4em
\item $\lambda_1^{(\bm{x}^{\mathrm{ex}})}/\mu_1^{(\bm{x}^{\mathrm{ex}})}>\lambda_2^{(\bm{x}^{\mathrm{ex}})}/(\mu_2^{(\bm{x}^{\mathrm{ex}})}\beta_E)$ or $\lambda_2^{(\bm{x}^{\mathrm{ex}})}/\mu_2^{(\bm{x}^{\mathrm{ex}})}>\lambda_1^{(\bm{x}^{\mathrm{ex}})}/(\mu_1^{(\bm{x}^{\mathrm{ex}})}\beta_E)$, if $e_1=e_2=0$;

\item $\lambda_1^{(\bm{x}^{\mathrm{ex}})}/\mu_1^{(\bm{x}^{\mathrm{ex}})}\neq\lambda_2^{(\bm{x}^{\mathrm{ex}})}/(\mu_2^{(\bm{x}^{\mathrm{ex}})}\beta_E)$, if $e_1>0$;

\item $\lambda_2^{(\bm{x}^{\mathrm{ex}})}/\mu_2^{(\bm{x}^{\mathrm{ex}})}\neq\lambda_1^{(\bm{x}^{\mathrm{ex}})}/(\mu_1^{(\bm{x}^{\mathrm{ex}})}\beta_E)$, if $e_2>0$.
\end{itemize}
 \end{proposition}
 \begin{proof}
See Appendix \ref{appendix:proof2}.
 \end{proof}

 \begin{remark}\label{theorem:remark}
Proposition \ref{theorem1} can be intuitively explained as follows by taking $\lambda_1^{(\bm{x}^{\mathrm{ex}})}/\mu_1^{(\bm{x}^{\mathrm{ex}})}>\lambda_2^{(\bm{x}^{\mathrm{ex}})}/(\mu_2^{(\bm{x}^{\mathrm{ex}})}\beta_E)$ when $e_1=e_2=0$ as an example. Other cases can be understood by similar observations. When $e_1=e_2=0$, this condition of $\lambda_1^{(\bm{x}^{\mathrm{ex}})}/\mu_1^{(\bm{x}^{\mathrm{ex}})}>\lambda_2^{(\bm{x}^{\mathrm{ex}})}/(\mu_2^{(\bm{x}^{\mathrm{ex}})}\beta_E)$ implies that we can always find $\Delta\mv{x}^{\mathrm{ex}}=[\Delta e_1, \Delta e_2, \Delta w_1, \Delta w_2]^T$ sufficiently small with $\Delta e_1>0, \Delta e_2=0, \Delta w_1=0$ and $\Delta w_2>0$ such that $\lambda_1^{(\bm{x}^{\mathrm{ex}})}/\mu_1^{(\bm{x}^{\mathrm{ex}})}>\Delta e_1/\Delta w_2>\lambda_2^{(\bm{x}^{\mathrm{ex}})}/(\mu_2^{(\bm{x}^{\mathrm{ex}})}\beta_E)$. In other words, there exists a new joint energy and spectrum cooperation scheme for the  costs of both systems to be reduced at the same time, i.e., $\nabla \bar{C}_1(\mv{x}^{\mathrm{ex}})^T\Delta\mv{x}^{\mathrm{ex}}={\mu_1^{(\bm{x}^{\mathrm{ex}})}}{\Delta e_1}-{\lambda_1^{(\bm{x}^{\mathrm{ex}})}}{\Delta w_2}<0$ and $\nabla \bar{C}_2(\mv{x}^{\mathrm{ex}})^T\Delta\mv{x}^{\mathrm{ex}}={\lambda_2^{(\bm{x}^{\mathrm{ex}})}}{\Delta w_2}-{\mu_2^{(\bm{x}^{\mathrm{ex}})}}\beta_E{\Delta e_1} <0 $. By using the marginal cost interpretation in Proposition \ref{proposition:5.1},  the costs at both BSs can be further reduced by transferring  ${\Delta e_1}$ amount of energy from BS 1 to BS 2 and sharing ${\Delta w_2}$ amount of spectrum from BS 2 to BS 1.\end{remark}

Finally, we can characterize the conditions for partial cooperation by examining $\mv{x}^{\mathrm{ex}} = \mv{0}$ in Proposition \ref{theorem1}. We explicitly give the conditions as follows.

\begin{corollary}\label{condition}
Partial cooperation is feasible if and only if $\lambda_1^{(\bm{0})}/\mu_1^{(\bm{0})}>\lambda_2^{(\bm{0})}/(\mu_2^{(\bm{0})}\beta_E)$  or  $\lambda_2^{(\bm{0})}/\mu_2^{(\bm{0})}>\lambda_1^{(\bm{0})}/(\mu_1^{(\bm{0})}\beta_E)$.
\end{corollary}

 Corollary \ref{condition} is implied by Proposition \ref{theorem1}. More intuitively, under the condition of $\lambda_1^{(\bm{0})}/\mu_1^{(\bm{0})}>\lambda_2^{(\bm{0})}/(\mu_2^{(\bm{0})}\beta_E)$, it follows from Remark \ref{theorem:remark} that BS 1 is more spectrum-hungry than BS 2, while BS 2 is more insufficient of energy than BS 1.  Hence, the costs at both BS can be reduced at the same time by BS 1 transferring spectrum to BS 2 and BS 2 transferring energy to BS 1. Similarly, if $\lambda_2^{(\bm{0})}/\mu_2^{(\bm{0})}>\lambda_1^{(\bm{0})}/(\mu_1^{(\bm{0})}\beta_E)$, the opposite is true. This shows that partial cooperation is only feasible when two systems find inter-system complementarity in energy and spectrum resources.

\begin{figure*}
        \centering
        \begin{subfigure}[b]{0.5\textwidth}
                \centering
                \includegraphics[width=8cm]{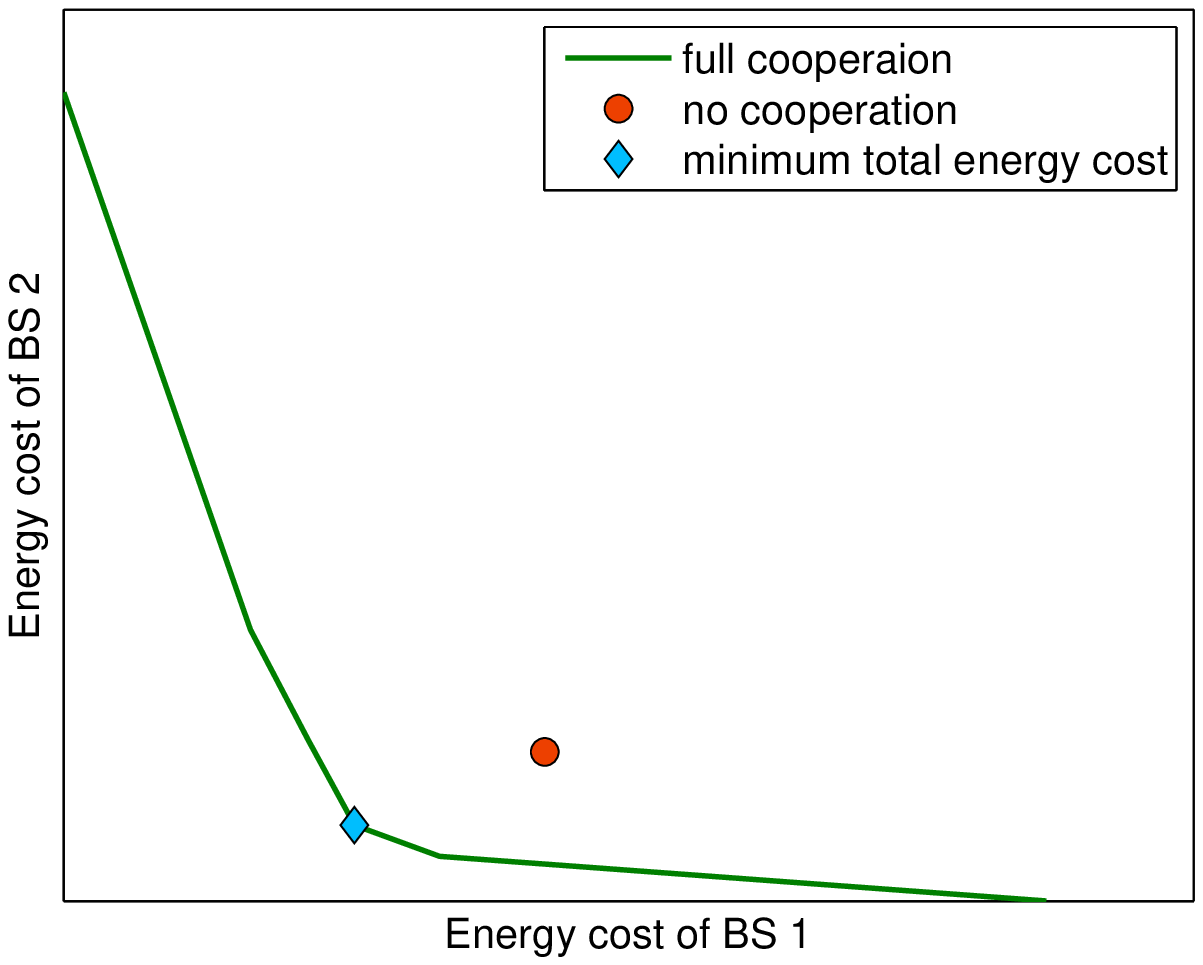}
                \caption{Partial Cooperation Feasible Scenario.}
             \label{Fully_Cooperative1}
        \end{subfigure}%
        \begin{subfigure}[b]{0.5\textwidth}
                \centering
                \includegraphics[width=8cm]{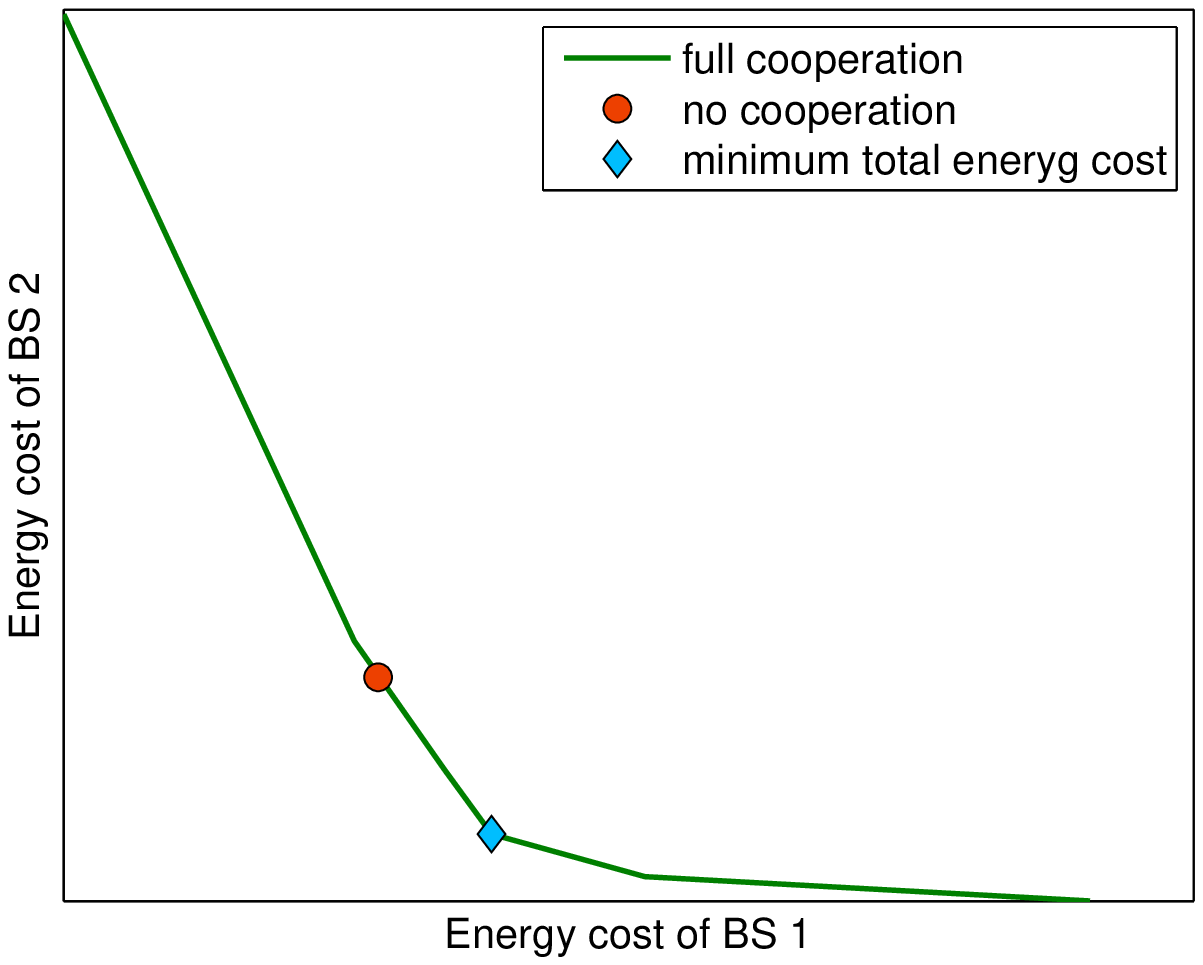}
                \caption{Partial Cooperation Infeasible Scenario.}
                \label{Fully_Cooperative2}
        \end{subfigure}
        \caption{Two different scenarios with joint energy and spectrum cooperation.}\label{fig:cooperation}
\end{figure*}

\begin{example}
We provide an example in Fig. \ref{fig:cooperation} to illustrate partial cooperation conditions  in Corollary \ref{condition}. We plot the Pareto boundary achieved by full cooperation, non-cooperation benchmark and the point corresponding to the minimum total cost (full cooperation with $\gamma_1=\gamma_2=1$). The joint energy and spectrum cooperation results in two scenarios as shown in Figs. \ref{Fully_Cooperative1} and \ref{Fully_Cooperative2}, which correspond to cases where partial cooperation is feasible and infeasible, respectively.
\begin{itemize}
\item Fig. \ref{Fully_Cooperative1} shows the feasible partial cooperation scenario, in which the partial cooperation conditions are satisfied. In this scenario, the non-cooperative benchmark is observed to lie within the Pareto boundary of cost region. As a result, from the non-cooperative benchmark, the costs of both BSs can be reduced at the same time until reaching the Pareto boundary.

\item Fig. \ref{Fully_Cooperative2} shows the scenario when the partial cooperation conditions are not satisfied,  where the non-cooperation benchmark is observed to lie on the Pareto boundary. From this result, it is evident that the two BSs' costs cannot be reduced at the same time. That is, the partial cooperation is infeasible.

\item In both scenarios of Figs. \ref{Fully_Cooperative1} and \ref{Fully_Cooperative2}, it is observed that the minimum total cost point differs from the non-cooperation benchmark. This shows that full cooperation can decrease the total cost at two BSs from the non-cooperative benchmark even when partial cooperation is infeasible, which can be realized by uni-directional cooperation (e.g., in Fig. \ref{Fully_Cooperative2}).

    \end{itemize}
 The results in this example motivate us to propose distributed algorithms for the partial cooperation scenario to reduce two BSs' costs from non-cooperative benchmark to Pareto optimality, as will be discussed next.
\end{example}

\subsection{Distributed Algorithm}
In this subsection, we design a distributed algorithm to implement the energy and spectrum cooperation for two partially cooperative systems (satisfying Corollary \ref{condition}).\footnote{As long as each system agrees to install the algorithm to benefit from its efficiency and fairness, the system will not make any deviation in its decisions as the algorithm runs automatically.} Since the two systems are selfish, we need to ensure that they can  improve their performance fairly. We design our algorithm based on the proportionally fair cost reduction, which is defined as follows.
\begin{definition}\label{definition}
Proportional fair cost reduction is achieved by both systems
 if, for the resultant cost tuple $(\tilde{C}_1, \tilde{C}_2)$, the cost reduction ratio between two BSs equals the  ratio of their costs in the non-cooperative scenario, i.e.,
\begin{align}
 \frac{\bar C_1(\bm{0})-\tilde{C}_1}{\bar C_2(\bm{0})-\tilde{C}_2}=\frac{\bar C_1(\bm{0})}{\bar C_2(\bm{0})}.
\end{align}
\end{definition}

Next, we proceed to elaborate on the key issue of the update of the energy and spectrum cooperation decision vector $\mv{x}^{\mathrm{ex}}$ to have proportionally fair cost reductions.  Our algorithm begins with the non-cooperative benchmark (i.e., $\mv{x}^{\mathrm{ex}}=\mv{0}$). Then, the inter-system energy and spectrum cooperation is adjusted to decrease the costs at both BSs in each iteration. Specifically, under any given $\mv{x}^{\mathrm{ex}}$, if the conditions in Proposition \ref{theorem1} are satisfied, then the two BSs cooperate by updating their energy and spectrum cooperation decision vector according to\begin{align}
  {\mv{x}^{\mathrm{ex}}}'=\mv{x}^{\mathrm{ex}}+\delta\mv{d}\label{update},
\end{align}
where $\delta>0$ is a sufficiently small step size and  $\mv{d}\in\mathbb{R}^4$ is the direction of the update that satisfies $\nabla \bar{C}_i(\mv{x}^{\mathrm{ex}})^T\mv{d}<0$ (cf. (\ref{1storder})). It can be observed that there are multiple solutions satisfying this condition. Here, we choose $\mv{d}$ in each iteration as follows:
\begin{itemize}
  \item If $\lambda_1^{(\bm{0})}\mu_2^{(\bm{0})}\beta_E>\lambda_2^{(\bm{0})}\mu_1^{(\bm{0})}$ holds, which means that costs of both systems can be reduced by system 1 sharing energy to system 2 and system 2 sharing spectrum to system 1 (cf. Corollary \ref{condition}), then we choose \begin{align}\label{determination:d1}\mv{d}=&
      \mathrm{sign}\left(\lambda_1^{(\mv{x}^{\mathrm{ex}})}\mu_2^{(\mv{x}^{\mathrm{ex}})}\beta_E-\lambda_2^{(\mv{x}^{\mathrm{ex}})}\mu_1^{(\mv{x}^{\mathrm{ex}})}\right)\cdot \nonumber\\
      &\left[\rho\lambda_2^{(\bm{x}^{\mathrm{ex}})}+\lambda_1^{(\bm{x}^{\mathrm{ex}})}, 0,   0, \mu_1^{(\bm{x}^{\mathrm{ex}})}+\rho\beta_E\mu_2^{(\bm{x}^{\mathrm{ex}})} \right]^T.\end{align}
  \item  If $\lambda_2^{(\bm{0})}\mu_1^{(\bm{0})}\beta_E>\lambda_1^{(\bm{0})}\mu_2^{(\bm{0})}$ holds, which means that  costs of both systems can be reduced by system 1 sharing spectrum to system 2 and system 2 sharing energy to system 1 (cf. Corollary \ref{condition}), then we choose
  \begin{align}\label{determination:d2}\mv{d}=&\mathrm{sign}\left(\lambda_2^{(\bm{x}^{\mathrm{ex}})}\mu_1^{(\bm{x}^{\mathrm{ex}})}\beta_E-\lambda_1^{(\bm{x}^{\mathrm{ex}})}\mu_2^{(\bm{x}^{\mathrm{ex}})}\right)\cdot \nonumber\\
  &\times \left[0,\rho\lambda_2^{(\bm{x}^{\mathrm{ex}})}+\lambda_1^{(\bm{x}^{\mathrm{ex}})}, \mu_1^{(\bm{x}^{\mathrm{ex}})}+\rho\beta_E\mu_2^{(\bm{x}^{\mathrm{ex}})},0 \right]^T. \end{align}
\end{itemize}
Here,  $\mathrm{sign}(x)=1$ if $x\geq0$ and $\mathrm{sign}(x)=-1$ if $x<0$, and $\rho$ is a factor controlling the ratio of cost reduction at both BSs in each update.
With the choice of $\mv{d}$ as shown above, the decrease of cost for each BS in each update is
  \begin{align}\label{eqn:ratio}
   \left[
 \begin{array}{c}
\bar{C}_{1}({\mv{x}^{\mathrm{ex}}}') - \bar{C}_{1}(\mv{x}^{\mathrm{ex}})\\
\bar{C}_{2}({\mv{x}^{\mathrm{ex}}}') - \bar{C}_{2}(\mv{x}^{\mathrm{ex}})
 \end{array}
 \right]=
 \left[
 \begin{array}{c}
   \nabla \bar{C}_{1}(\mv{x}^{\mathrm{ex}})^T\mv{d}\\
   \nabla \bar{C}_2(\mv{x}^{\mathrm{ex}})^T\mv{d}
 \end{array}
 \right]=\sigma
  \left[
 \begin{array}{c}
   \rho\\
   1
 \end{array}
 \right],
  \end{align}
 where $\sigma \leq 0$ is obtained by substituting (\ref{update}) into (\ref{1storder}), given by
\begin{align}
   \sigma=
\begin{cases}
-(\lambda_1^{(\mv{x}^{\mathrm{ex}})}\mu_2^{(\mv{x}^{\mathrm{ex}})}\beta_E-\lambda_2^{(\mv{x}^{\mathrm{ex}})}\mu_1^{(\mv{x}^{\mathrm{ex}})}),\nonumber\\ ~~~~~~~~~~~~~~~~~~~~~~~~~~\lambda_1^{(\mv{x}^{\mathrm{ex}})}\mu_2^{(\mv{x}^{\mathrm{ex}})}\beta_E\geq\lambda_2^{(\mv{x}^{\mathrm{ex}})}\mu_1^{(\mv{x}^{\mathrm{ex}})} \\
-(\lambda_2^{(\mv{x}^{\mathrm{ex}})}\mu_1^{(\mv{x}^{\mathrm{ex}})}\beta_E-\lambda_1^{(\mv{x}^{\mathrm{ex}})}\mu_2^{(\mv{x}^{\mathrm{ex}})}),\nonumber\\
~~~~~~~~~~~~~~~~~~~~~~~~~~\lambda_2^{(\mv{x}^{\mathrm{ex}})}\mu_1^{(\mv{x}^{\mathrm{ex}})}\beta_E\geq\lambda_1^{(\mv{x}^{\mathrm{ex}})}\mu_2^{(\mv{x}^{\mathrm{ex}})} \\
0,~~~~~~~~~~~~~~~~~~~~~~~~\mathrm{otherwise}
\end{cases}.\nonumber
\end{align}
From (\ref{eqn:ratio}), it follows that the cost reduction in each iteration satisfies $\rho=\frac{\bar{C}_1({\mv{x}^{\mathrm{ex}}})-\bar{C}_1({\mv{x}^{\mathrm{ex}}}')}{\bar{C}_2({\mv{x}^{\mathrm{ex}}})-\bar{C}_2({\mv{x}^{\mathrm{ex}}}')}$. Using this fact together with the proportional fairness criterion in Definition 5.1,  $\rho$ is determined as
\begin{align}
\rho=\frac{\bar{C}_1(\bf{0})}{\bar{C}_2(\bf{0})}\label{rho}.
\end{align}

 \begin{remark}
 Generally, it follows from  (\ref{eqn:ratio}) that $\rho$ controls the ratio of cost reduction at the two BSs. Besides the proportionally fair choice of $\rho$ in  (\ref{eqn:ratio}), we can set other values of $\rho>1$ (or $\rho<1$) to ensure that a larger (or smaller) cost decrease is achieved for BS 1 compared to BS 2 (provided that the step size $\delta$ is sufficiently small). By exhausting $\rho$ from zero to infinity, we can achieve all points on the Pareto boundary that have lower costs at both BSs than the non-cooperative benchmark.
\end{remark}\label{importantremark2}

\begin{table}[htp]
\begin{center}
\caption{Distributed Algorithm for Partial Cooperation}
 \hrule\vspace{0.2cm} \textbf{Algorithm I}  \vspace{0.2cm}
\hrule \vspace{0.2cm}
\begin{itemize}
\item[a)] Each BS $i\in\{1,2\}$ initializes from the non-cooperative benchmark by setting $e_i = w_i = 0$ (i.e., $\bm{x}^{\mathrm{ex}}=\mv{0}$). Each BS $i$ solves the problem in (\ref{P4}) for obtaining $\lambda_i^{(\bm{0})}$ and $\mu_i^{(\bm{0})}$, and sends them to the other BS $\bar \iota$.

  \item[b)] Each BS $i\in\{1,2\}$ tests the conditions in Corollary \ref{condition}. If    $\lambda_1^{(\bm{0})}\mu_2^{(\bm{0})}\beta_E>\lambda_2^{(\bm{0})}\mu_1^{(\bm{0})}$, then choose $\mv{d}$ in (\ref{determination:d1}) as the the update vector in the following iterations. If      $\lambda_2^{(\bm{0})}\mu_1^{(\bm{0})}\beta_E>\lambda_1^{(\bm{0})}\mu_2^{(\bm{0})}$,  then choose $\mv{d}$ in (\ref{determination:d2}). Otherwise, the algorithm ends. Sets $\rho$ as  in (\ref{rho}).
\item[c)] {\bf Repeat:}
    \begin{itemize}
    \item[1)] Each BS $i\in\{1,2\}$ computes the dual variables $\lambda_i^{(\bm{x}^{\mathrm{ex}})}$ and $\mu_i^{(\bm{x}^{\mathrm{ex}})}$ by solving the problem in (\ref{P4}), and sends them to the other BS $\bar{\imath}$;
    \item[2)] BS $i\in\{1,2\}$ updates the energy and spectrum cooperation vector as ${\mv{x}^{ex}}'={\mv{x}^{ex}}+\delta\mv{d}$;
    \item[3)] ${\mv{x}^{ex}} \gets {\mv{x}^{ex}}'$.
    \end{itemize}
\item[d)] {\bf Until} the conditions in Proposition \ref{theorem1} are satisfied.
\end{itemize}
\vspace{0.2cm} \hrule \label{algorithm:2}
\end{center}
\end{table}

To summarize, the distributed algorithm for partial cooperation is presented in Table \ref{algorithm:2} as Algorithm I and is described as follows. Initially, each BS  $i\in\{1,2\}$ begins from the non-cooperation benchmark case with $\mv{x}^{\mathrm{ex}}=\mv{0}$ and determines the update vector $\mv{d}$ that will be used in each iterations. Specifically, each BS computes $\lambda_i^{(\bm{0})}$ and $\mu_i^{(\bm{0})}$, and shares them with each other. If $\lambda_1^{(\bm{0})}\mu_2^{(\bm{0})}\beta_E>\lambda_2^{(\bm{0})}\mu_1^{(\bm{0})}$, then choose $\mv{d}$ in (\ref{determination:d1}) as the update vector; while if $\lambda_2^{(\bm{0})}\mu_1^{(\bm{0})}\beta_E>\lambda_1^{(\bm{0})}\mu_2^{(\bm{0})}$, then choose $\mv{d}$ in (\ref{determination:d2}). We set the cost reduction ratio as in (\ref{rho}). Then, the following procedures are implemented iteratively. In each iteration, according to the current energy and spectrum cooperation vector $\mv{x}^{\mathrm{ex}}$, each BS computes the dual variables $\lambda_i^{(\bm{x}^{\mathrm{ex}})}$ and $\mu_i^{(\bm{x}^{\mathrm{ex}})}$ by solving the problem in (\ref{P4}) and sends them to the other BS.  After exchanging the dual variables, the two BSs examine the conditions in Proposition \ref{theorem1} individually. If the conditions are satisfied, then each BS updates the cooperation scheme $\mv{x}^{\mathrm{ex}}$ according to (\ref{update}). The procedure shall proceed until the two BSs cannot decrease their costs at the same time, i.e., conditions in Proposition \ref{theorem1} are not satisfied. Due to the fact that the algorithm can guarantee the costs to decrease proportionally fair at each iteration and the Pareto optimal costs are bounded, the algorithm can always converge to a Pareto optimal point with proportional fairness provided that the step size $\delta$ is sufficiently small. {  Note that Algorithm I minimizes both systems' costs simultaneously based on the gradients of two convex cost functions in (\ref{P4}), which differs from the conventional gradient descent method in convex optimization which minimizes a single convex objective \cite{Boyd04}.
}

Compared to the centralized joint energy and spectrum cooperation scheme, which requires a central unit to gather all channel and energy information at two systems, the distributed algorithm only needs the exchange of {  four scalers (i.e. the marginal spectrum and energy prices $\lambda_i^{(\bm{x}^{\mathrm{ex}})}$ and $\mu_i^{(\bm{x}^{\mathrm{ex}})}, \forall i\in\{1,2\} )$} between two BSs in each iteration. As a result, such distributed algorithm can preserve the two systems' privacy and greatly reduce the cooperation complexity (e.g., signaling overhead).

\section{Numerical Results}\label{sec:numerical_examples}

\begin{figure}[t]
  \centering
  \includegraphics[width=8cm]{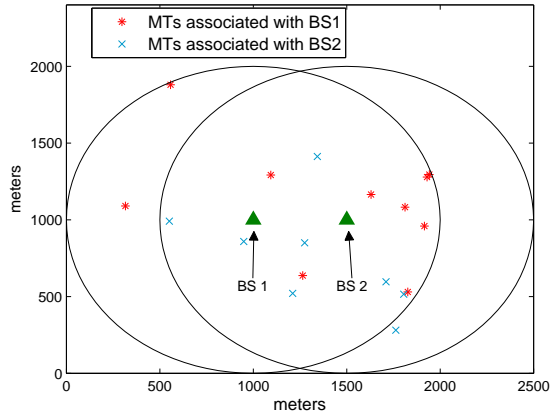}\\
  \caption{Simulation Setup.}\label{simu_setup}
\end{figure}
In this section, we provide numerical results for evaluating the performance of our proposed joint energy and spectrum cooperation. For the simulation setup, we assume that BS 1 and BS 2 each covers a circular area with a radius of 500 meters (m) as shown in Fig. \ref{simu_setup}. $K_1=10$ and $K_2=8$ MTs are randomly generated in the two cells. We consider a simplified path loss model for the wireless channel with the channel gain set as $g_k=c_0(\frac{d_k}{d_0})^{-\zeta}$, where $c_0=-60$dB is a constant path loss at the reference distance $d_0=10$ m, $d_k$ is the distance between MT $k$ and its associated BS in meter and $\zeta=3$ is the path loss exponent. The noise PSD at each MT is set as $N_0=-150$ dBm/Hz. Furthermore, we set the non-transmission power consumption for the BSs as $P_{c,1}=P_{c,2}=100$ Watts(W). The maximum usable renewable energy at the two BSs are $\bar E_1=190$ W and $\bar E_2=130$ W, respectively. We set the energy price from renewable utility firm and power grid as $\alpha_{i}^E=0.2$/W and $\alpha_{i}^G=1$/W, respectively, where the price unit is normalized for simplicity. The bandwidth for the two BSs are  $W_1=15$ MHz and $W_2=20$ MHz, respectively.

\begin{figure}[t]
  \centering
  \includegraphics[width=9cm]{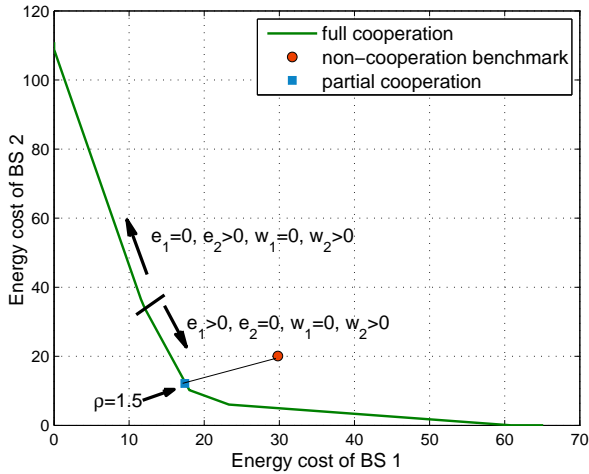}\\
  \caption{Energy cost region for the case of joint energy and spectrum cooperation versus the case without energy or spectrum cooperation.}\label{fullcoop}
\end{figure}

Fig. \ref{fullcoop} shows the BSs' optimized costs by the proposed joint energy and spectrum cooperation in full cooperation with $\beta_B=1$ and $\beta_E=0.8$ compared with the  non-cooperation benchmark. Notice that Fig. \ref{fullcoop} only shows a cooperation case at a time slot that BS 2 has relatively more bandwidth (considering its realized traffic load) than BS 1 and the spectrum cooperation is from BS 2 to BS 1. Yet, in other time slots two different BSs' traffic loads and channel realizations can change and their spectrum cooperation may follow a different direction.  It is observed that the non-cooperation benchmark lies within the Pareto boundary achieved by full cooperation, while the partial cooperation lies on that Pareto boundary. This indicates the benefit of joint energy and spectrum cooperation in minimizing the two systems' costs. It is also observed that the Pareto boundary in full cooperation is achieved by either uni-directional cooperation with BS 2 transferring both energy and spectrum to BS 1 (i.e., $e_1=0, e_2>0$ and $w_1=0, w_2>0$) or bi-directional cooperation with BS 1 transferring energy to BS 2 and BS 2 transferring spectrum to BS 1 (i.e., $e_1>0, e_2=0$ and $w_1=0, w_2>0$). Specifically, the energy costs at both BSs are decreased simultaneously  compared to the non-cooperation benchmark only in the case of bidirectional cooperation. This is intuitive, since otherwise the cost of the  BS that shares both resources will increase. Furthermore, for our proposed distributed algorithm, it is observed that it converges to the proportional fair result with $\rho=1.5$, which lies on the Pareto boundary in this region. This is also expected, since partial cooperation is only feasible when both systems find complementarity between energy and spectrum resources.

\begin{figure}[t]
  \centering
  \includegraphics[width=8cm]{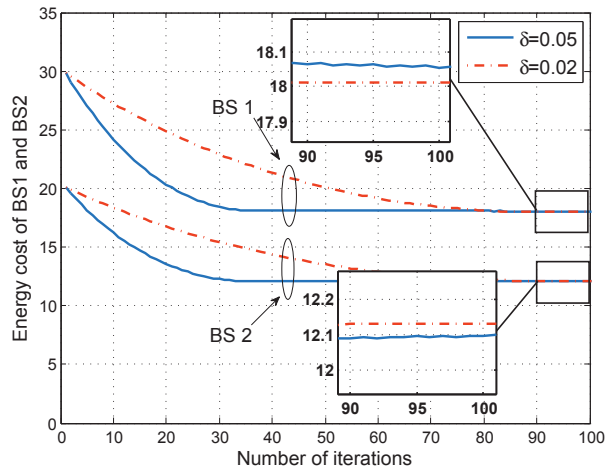}\\
  \caption{Convergence of the distributed algorithm under different step-size $\delta$'s.}\label{delta}
\end{figure}
\begin{figure}[t]
  \centering
  \includegraphics[width=8.5cm]{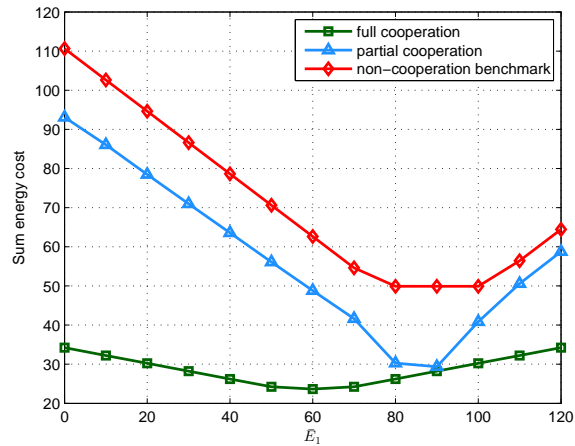}\\
  \caption{Comparison of the sum energy cost under different $\bar{E}_i$'s with $\bar{E}_1+\bar{E}_2=120$ W.}\label{compare}
\end{figure}

In Fig. \ref{delta}, we show the convergence of the partially cooperative distributed algorithm under step-sizes $\delta=0.05$ and $\delta=0.02$, and  $\rho=1.5$ is chosen to achieve the proportional fairness. It is observed that under different step-sizes, the costs at two BSs converge to different points on the Pareto boundary. Specifically, given $\delta = 0.05$,  the cost reductions at BS 1 and BS 2 are observed to be 11.7606 (from 29.8092 to 18.0423) and 7.9825 (from 20.0860 to 12.1035), respectively, with the cost reduction ratio being 11.7606/7.9825=1.4733;  while given $\delta = 0.02$,  the cost reductions at BS 1 and BS 2 are observed to be 11.7968(from 29.8092 to 18.0124) and 7.9625 (from 20.0860 to 12.1235), respectively, with the cost reduction ratio being 11.7968/7.9625=1.4815. By comparing the cost reduction ratios in two cases with $\rho = 1.5$, it is inferred that the proportional fairness can be better guaranteed with smaller $\delta$. This also validates that $\delta$ should be sufficiently small to ensure the proportionally cost reduction in each iteration step (see Section V).  It is also observed that the algorithm converges after about 40 iterations for $\delta=0.05$. This indicates that under proper choice of $\delta$, the convergence speed is very fast provided that certain proportional fairness inaccuracies are admitted.

In Fig. \ref{compare}, we compare the achieved total costs of two BSs by different schemes (i.e., $C_1+C_2$ with weights $\gamma_1 = \gamma_2 = 1$) versus  the renewable energy level at BS 1 (i.e., $\bar E_1$) subject to  $\bar E_1 + \bar E_2 = 120$ W. From this figure, it is observed that the fully cooperative scenario outperforms both the non-cooperative benchmark and the partially cooperative scenario, especially when the available renewable energy amounts at two BSs are not even (e.g., $\bar E_1 = 0$ W and $\bar E_2 = 120$ W, as well as, $\bar E_1  = 120$ W and $\bar E_2 = 0$ W). This is intuitive, since in this case, the partially cooperative systems may have limited incentives for cooperation, while the fully cooperative systems can implement uni-directional energy and spectrum cooperation to reduce the total cost. It is also observed that the performance gap between fully and partially cooperative scenario becomes smallest when $\bar E_1 = 90$ W. This is because the partially cooperative scenario can find the best complementarity to bidirectionally exchange the two resources as in full cooperation scenario. More specifically, in this case, the cost reduction ratio between two BSs for the fully cooperative scheme  is most close to the cost ratio between two BSs in the non-cooperative benchmark.  Hence,  full cooperation in this case results in an inter-system energy and spectrum cooperation scheme that is consistent with the proportionally fair criterion in partial cooperation (see Definition \ref{definition}).

\begin{figure*}[t]
        \centering
        \begin{subfigure}{0.5\textwidth}
                \centering
                \includegraphics[width=8cm]{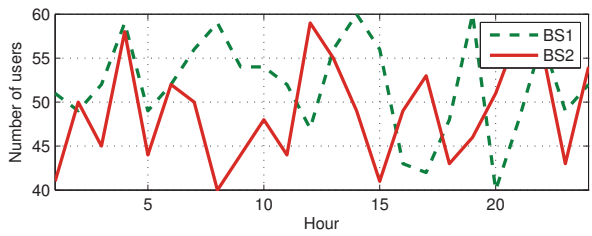}
                \caption{Number of users $K_1$ and $K_2$.}
             \label{bs1_final}
        \end{subfigure}%
        \begin{subfigure}{0.5\textwidth}
                \centering
                \includegraphics[width=8cm]{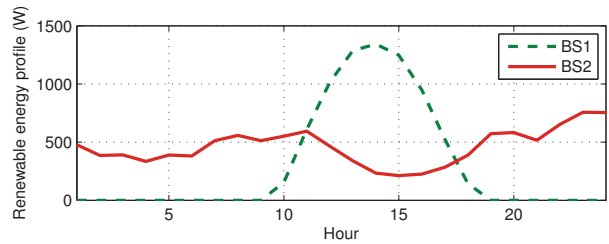}
                \caption{Renewable energy profile $\bar E_1$ and $\bar E_2$.}
                \label{bs2_final}
        \end{subfigure}
  \caption{Traffic and harvested renewable energy profile at the two systems for simulation.}\label{wind_solar_energy_compare}
\end{figure*}
\begin{figure*}[t]
        \centering
        \begin{subfigure}{0.5\textwidth}
                \centering
                \includegraphics[width=8cm]{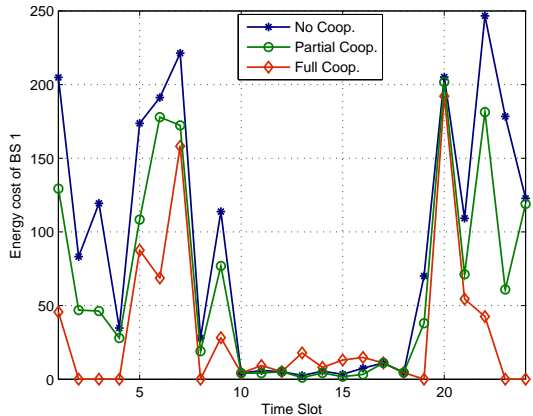}
               \caption{Energy cost of BS 1}
             \label{bs1_cost_final}
        \end{subfigure}%
        \begin{subfigure}{0.5\textwidth}
                \centering
                \includegraphics[width=8cm]{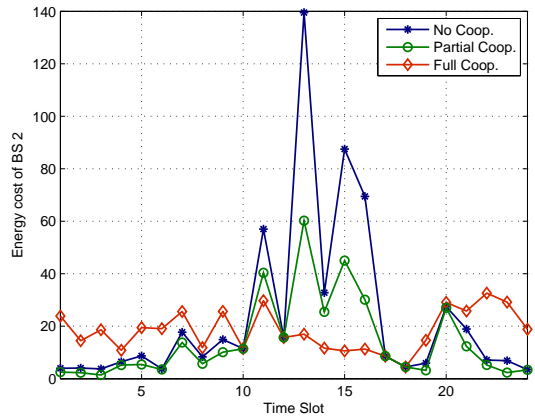}
                \caption{Energy cost of BS 2}
                \label{bs2_cost_final}
        \end{subfigure}
        \caption{Energy costs of the two BSs under full, partial and no cooperation.}\label{fig:final}
\end{figure*}

{
Finally, we show the optimized costs at both BSs over time by considering stochastically varying traffic load and harvested renewable energy. We assume that BS 1 and BS 2 are powered by solar and wind energy, respectively, and their energy harvesting rates are based on the  real-world solar and wind data from Elia, a Belgium electricity transmission system operator.\footnote{See http://www.elia.be/en/grid-data/power-generation/} For demonstration, we use the average harvested energy over one hour as $\bar E_i$ at each slot, as shown in Fig. \ref{bs1_final}, and thus our studied 24 slots correspond to the energy harvesting profile over one day. Furthermore, we consider the number of MTs served by each BS, $K_i$'s, over slots as shown in Fig. \ref{bs2_final}, which are randomly generated based on a discrete uniform distribution over the interval [40, 60] for simplification. Under this setup, Fig. \ref{fig:final} shows the optimized costs of the two BSs by different schemes. It is observed that over the 24 slots the full and partial cooperation achieve 55.68\% and 33.75\% total cost reduction for the two BSs as compared to the no cooperation benchmark, respectively. It is also observed that for partial cooperation, the energy costs of both BSs are reduced at the same time, while for full cooperation, at certain slots, i.e., slots 22,  the cost of one BS is reduced significantly at the expense of the cost increase of the other BSs. The reason is as follows. In partial cooperation, the two systems seek for mutual benefits and thus only bidirectional cooperation is feasible. For instance, at time slot  6, BS 1 shares bandwidth to BS 2 and BS 2 transfers energy to BS 1.  In contrast, under full cooperation, it is not required that the costs of both BSs be reduced at the same time and the common goal is to reduce the sum energy costs. Hence, in this case, it is possible that uni-directional cooperation happens where one BS sacrifices its interest to another one in order to reduce the sum energy cost.}

\section{Conclusions}
In this paper, we propose a joint energy and spectrum cooperation approach to reduce the energy costs at two wireless cellular systems that are powered by both energy harvesting and power grid. We minimize the costs at both systems by considering two scenarios where the wireless systems belong to the same entity and different entities, respectively. In the former case with full cooperation, we propose an optimal centralized algorithm for achieving the minimum weighted sum energy cost at two BSs. In the latter case with partial cooperation, we develop a distributed algorithm to achieve the Pareto optimal energy costs with proportional fairness. Our results provide insights on the design of  cooperative cellular systems with both energy and spectrum cooperation. Nevertheless, due to the space limitation, there are still several important issues on joint energy and spectrum cooperation that are not addressed in this paper, some of which are briefly discussed as follows to motivate future work:
\begin{itemize}
  \item { In this paper, we consider energy cooperation without storage at BSs. However, when energy storage is implemented in the systems, the energy management of the BS can have more flexibility such that  the energy supply variations in time can be mitigated and energy cooperation between the two systems can more efficiently exploit the geographical energy diversity. For example, at any given slot, a BS with sufficient renewable energy can either share them to other BSs with insufficient renewable energy, or store them for future use.  However, under the setup with storage, the joint optimization of the energy and spectrum cooperation over space and the storage management over time requires the prediction of the energy price, renewable energy availability and the user traffic in the future. The design will be a stochastic dynamic programming problem, whose optimal solution is still unknown and worth pursuing.}
  \item In this work, we consider that the two systems operate over orthogonal frequency bands. In general, allowing MTs associated with different BSs to share the same frequency band may further improve the spectrum efficiency. However, this formulation will  turn the problem into a challenging one related to the interference channel. The optimal resource allocation scheme in this case is unknown and difficult to solve.
   \item  We  have discussed the joint energy and spectrum sharing for energy saving under the setup of hybrid energy supply. In addition to this, another interesting work direction can be the maximization of QoS performance with energy cooperation {subject to the resource (i.e., spectrum and power) constraints. Depending on different application scenarios, the QoS metrics can be delay \cite{TianZhang2013}, throughput \cite{HoZhang2012}, and etc.} The details on the optimal strategies in these scenarios can be modeled and solved similarly as in this paper.
\end{itemize}

\appendices
\section{Proof of Proposition \ref{prop_nocoop}}\label{no_coop_prop}

First, we obtain the optimal bandwidth allocation $\{b_k^{\star}\}$  and optimal power allocation $\{p_k^{\star}\}$. Given that the non-transmission power at BS $i$, $P_{c,i}$, is constant, it can be shown that the objective of (P2), i.e., the cost at BS $i$, is a monotonically { increasing} function of the sum transmission power $\sum_{k\in\mathcal{K}_i}p_k $ at BS $i$, no matter the power is purchased from the conventional grid or the renewable utility firm.
 Thus, deriving the optimal $\{b_k^{\star}\}$ and $\{p_k^{\star}\}$ to (P2) is equivalent to minimizing $\sum_{k\in\mathcal{K}_i}p_k $ at BS $i$ subject to the bandwidth constraint in (\ref{p1cons3:nocoop}) and the QoS constraints in (\ref{p1cons4:nocoop}). Using this argument together with (\ref{energy:nocoop}), we can obtain the optimal bandwidth allocation for (P2) by solving the following problem:
\begin{align}
  \mathop{\mathtt{min.}}_{\{b_k\ge0\}} &
~~ \sum_{k\in\mathcal{K}_i}\frac{b_kN_0}{g_k}\bigg(2^{\frac{r_k}{b_k}}-1\bigg) \nonumber\\
\mathtt{s.t.} & ~~ \sum_{k\in\mathcal{K}_i}b_k\leq W_i. \label{simplifiednocoop}
\end{align}

Since problem (\ref{simplifiednocoop}) is convex and satisfies the Slater's condition \cite{Boyd04}, the KKT conditions given as follows are necessary and sufficient for its optimal solution.

\begin{align}
 &\frac{N_0}{g_k}\bigg(2^{\frac{r_k}{b_k}}-1\bigg)-\frac{N_0r_k}{g_kb_k}\ln2\cdot2^{\frac{r_k}{b_k}}+\nu_i-\zeta_k=0
  ,~\forall k\in\mathcal{K}_i,\label{kkt1}\\
 & \nu_i(\sum_{k\in\mathcal{K}_i}b_k-W_i)=0,~
  \nu_i\ge0,\label{kkt2}\\
 & \zeta_kb_k=0,~\zeta_k\geq0,~b_k\geq 0,~\forall k\in\mathcal K_i, \label{kkt3}
\end{align}

{where $\nu_i\geq0$ is the dual variable associated with the bandwidth constraint in (\ref{simplifiednocoop}) and $\zeta_k\geq 0$ is the dual variable for $b_k\ge0,~k\in\mathcal{K}_i$. Note that  the optimal bandwidth allocation $\{b_k^{\star}\}$ should satisfy that $b_k^{\star} > 0$, $\forall k\in\mathcal K_i$, otherwise the objective value in (\ref{simplifiednocoop}) will go to infinity, given the fact that $r_k > 0, \forall k \in \mathcal K_i$. By using this together with (\ref{kkt3}), we thus have $\zeta_k^{\star} = 0,\forall k\in\mathcal K_i$. Accordingly, it follows from (\ref{kkt1}) that ${b_k^{\star}}$ can be obtained as in (\ref{optimal_spectrum}), where $\nu_i^{\star}>0$ is determined by the equation $\sum_{k\in\mathcal{K}_i}b_k^{\star}= W_i$. Furthermore, by substituting the derived $\{b_k^{\star}\}$ into (\ref{energy:nocoop}), the optimal $\{p_k^{\star}\}$ can be obtained.}

Next, with $\{p_k^{\star}\}$ at hand, we proceed to obtain the optimal energy allocation  $E_i^{\star}$ and $G_i^{\star}$. By using $\alpha_i^E < \alpha_i^G$ together with the fact that the optimal solution of (P2) is attained when the power constraint in (\ref{p1cons1:nocoop}) is tight, it can be verified that the optimal solutions of $E_i^{\star}$ and $G_i^{\star}$ are obtained as in Proposition
\ref{prop_nocoop}. Hence, the proof of Proposition \ref{prop_nocoop} is complete.

\section{Proof of Proposition \ref{optimalP1}}\label{proofoptimalP1}

By substituting (\ref{Lambertx:nocoop}) into the power constraint (\ref{P1:cons1}) in (P1), we can re-express (P1) as
{\small
\begin{align}
(\mathrm{P1.1}):\mathop{\mathtt{min.}}_{\mv{x}\ge \bf{0}}&~
 \sum_{i=1}^{2}\gamma_i(\alpha^{E}_iE_i+\alpha^{G}_iG_{i})\nonumber\\
\mathtt{s.t.}
&~  \sum_{k\in\mathcal{K}_i}\frac{b_kN_0}{g_k}\bigg(2^{\frac{r_k}{b_k}}-1\bigg)+P_{c,i}\nonumber\\
&~\le E_i+G_i+\beta_{E}e_{\bar{\imath}}-e_i,~i \in \{1,2\}\label{P1:cons1:re},\\
&~(\rm{\ref{P1:cons3}})~\mathrm{and}~(\rm{\ref{P1:cons2}}).\nonumber
\end{align}}

Denote the dual variables associated with the constraints in (\ref{P1:cons1:re}) and (\ref{P1:cons3}) as $\mu_i\ge 0$ and $\lambda_i\geq 0,~ i \in \{1,2\}$, respectively. The { partial} Lagrangian of (P1.1) is then expressed as:
\begin{align}\label{app:lagrangian}
  \mathcal{L}(\mv{x},\{\mu_i\},&\{\lambda_i\})
=\sum_{i=1}^{2}E_{i}(\gamma_i\alpha^{E}_i-\mu_i)+\sum_{i=1}^{2}G_{i}(\gamma_i\alpha^{G}_i-\mu_i)\nonumber\\
  &+\sum_{i=1}^{2}\lambda_i\sum_{k\in\mathcal{K}_i}^{2}b_{k}+\sum_{i=1}^{2}\mu_{i}
  \sum_{k\in\mathcal{K}_i}^{2}\frac{b_kN_0}{g_k}
  \bigg(2^{\frac{r_k}{b_k}}-1\bigg)\nonumber\\
  &-\sum_{i=1}^{2}\lambda_iW_i+\sum_{i=1}^{2}\mu_iP_{c,i}
  +\sum_{i=1}^{2}w_{i}(\lambda_i-\beta_B\lambda_{\bar{\imath}})\nonumber\\
  &+\sum_{i=1}^{2}e_{i}(\mu_i-\beta_{E}\mu_{\bar{\imath}}).
\end{align}
Accordingly, the dual function can be obtained as
\begin{align}
  g(\{\mu_i\},\{\lambda_i\})=\mathop{\mathtt{min.}}_{\mv{x}\geq0}&~~\mathcal{L}(\mv{x}
  ,\{\mu_i\},\{\lambda_i\})\label{lagrangian}\\
  \mathtt{s.t.} &~~(\rm{\ref{P1:cons2}})\nonumber.
\end{align}
 Thus, the dual problem is expressed as
\begin{align}
  \mathrm{(P1.1-D)}:~\mathop{\mathtt{max.}}_{\{\mu_i\},\{\lambda_i\}}&~~g(\{\mu_i\},\{\lambda_i \}) \nonumber\\
\mathtt{s.t.} &~~\lambda_i\geq0,\mu_i\geq0,~\forall i \in \{1,2\}.\nonumber
\end{align}
Since (P1.1) is convex and satisfies the Slater's condition, strong duality holds between (P1.1) and (P1.1-D) \cite{Boyd04}. Therefore, (P1.1) can be solved optimally  by solving its dual problem (P1.1-D) as follows. We first solve the problem in (\ref{lagrangian}) to obtain $g(\{\mu_i\},\{\lambda_i\})$ for given $\{\mu_i\}$ and $\{\lambda_i\}$, and  then maximize $g(\{\mu_i\},\{\lambda_i\})$ over $\{\mu_i\}$ and $\{\lambda_i\}$.

We first give the following lemma.
\begin{lemma}\label{Lemma}
In order for $g(\{\mu_i\},\{\lambda_i\})$ to be bounded from below, it follows that
\begin{align}\gamma_i\alpha^{G}_i\geq \mu_i, ~\beta_E\mu_{\bar{\imath}}\leq\mu_i,~\beta_B\lambda_{\bar{\imath}}\leq\lambda_{i},~ \forall i \in \{1,2\}.
\label{eqn:bounded:below}\end{align}
\end{lemma}
\begin{proof}
First, suppose that $\gamma_i\alpha^{G}_i< \mu_i$ for any $i \in \{1,2\}$.
In this case, it is easy to verify that the dual function $g(\{\mu_i\},\{\lambda_i\})$ will go to minus infinity as $G_i\rightarrow \infty$, i.e., $g(\{\mu_i\},\{\lambda_i\})$ is unbounded from below. Hence, $\gamma_i\alpha^{G}_i\geq \mu_i,~i \in \{1,2\},$  should always hold.

Second,  suppose that $\beta_E\mu_{\bar{\imath}}>\mu_i$ for any $i \in \{1,2\}$. In this case, it is easy to verify that the dual function $g(\{\mu_i\},\{\lambda_i\})$ will go to minus infinity as $e_i\rightarrow \infty$, i.e., $g(\{\mu_i\},\{\lambda_i\})$ is unbounded from below. Hence, $\beta_E\mu_{\bar{\imath}}\leq\mu_i,~i \in \{1,2\},$ should always hold.

Last, suppose that $\beta_B\lambda_{\bar{\imath}}>\lambda_{i}$ for any $i \in \{1,2\}$. In this case, it is easy to verify that the dual function $g(\{\mu_i\},\{\lambda_i\})$ will go to minus infinity as $w_i\rightarrow \infty$, i.e., $g(\{\mu_i\},\{\lambda_i\})$ is unbounded from below. Hence, $\beta_B\lambda_{\bar{\imath}}\leq\lambda_{i},~i \in \{1,2\},$ should always hold.

By combining the above three arguments, Lemma \ref{Lemma} is thus proved.
\end{proof}

From Lemma \ref{Lemma}, it follows that the optimal solution of (P1.1-D) is achieved when $\{\mu_{i}\}$ and $\{\lambda_i\}$ satisfies the inequalities in (\ref{eqn:bounded:below}). As a result, we only need to solve problem (\ref{lagrangian}) with given $\{\mu_{i}\}$ and $\{\lambda_i\}$ satisfying (\ref{eqn:bounded:below}). In this case, it can be observed that problem (\ref{lagrangian}) can be decomposed into the following subproblems:
\begin{align}
  &\mathop{\mathtt{min.}}_{b_{k}\ge0}~~
  \lambda_ib_{k}+\mu_{i}
  \bigg(\frac{b_kN_0}{g_k}\bigg(2^{\frac{r_k}{b_k}}-1\bigg)\bigg),~k\in\mathcal{K}_1\cup\mathcal{K}_2 \label{sub1},\\
  &\mathop{\mathtt{min.}}_{0\le E_i\le \bar{E}_i}~~E_{i}(\gamma_i\alpha^{E}_i-\mu_i),~i \in \{1,2\}\label{sub2},\\
 & \mathop{\mathtt{min.}}_{G_{i}\geq0}~~G_{i}(\gamma_i\alpha^{G}_i-\mu_i),~i \in \{1,2\}\label{sub3},\\
 & \mathop{\mathtt{min.}}_{e_{i}\geq0}~~e_{i}(\mu_i-\beta_{E}\mu_{\bar{\imath}}),~i \in \{1,2\}\label{sub4},\\
 & \mathop{\mathtt{min.}}_{w_{i}\geq0}~~w_{i}(\lambda_i-\beta_B\lambda_{\bar{\imath}}),~i \in \{1,2\}.\label{sub5}
\end{align}

For the $K_1+K_2$ subproblems in (\ref{sub1}), the optimal bandwidth allocation with given $\{\mu_i\}$ and $\{\lambda_i\}$ can be obtained based on the first order condition and is expressed as
\begin{align}
  b_{k}^{(\mu_i,\lambda_i)}=
    \begin{cases}
\frac{r_k\ln 2}{\mathcal{W}(\frac{1}{e}(\frac{\lambda_i g_k}{\mu_iN_0}-1))+1}, & \mu_i>0 \\
0, & \mu_i=0\
\end{cases},~k\in \mathcal{K}_{i}, ~i \in \{1,2\}.
  \label{result1}
\end{align}
 Furthermore, the optimal solution to the subproblems in (\ref{sub2})-(\ref{sub5}) can be obtained as follows.
\begin{align}
  E_i^{(\mu_i)}&=
  \begin{cases}
0, & \gamma_i\alpha^{E}_i\geq \mu_i \\
\bar{E}_i, & \gamma_i\alpha^{E}_i< \mu_i\label{result2}
\end{cases}, ~i \in \{1,2\},\\
G_i^{(\mu_i)}&=0, ~i \in \{1,2\}\label{result3},
\\
e_i^{(\mu_i)}&=0, ~i \in \{1,2\}\label{result4},\\
w_i^{(\lambda_i)}&=0, ~i \in \{1,2\}.\label{result5}
\end{align}
{ Note that for the subproblems in (\ref{sub2}) with any $i \in \{1,2\}$, if $\gamma_i\alpha^{E}_i=\mu_i$, then the solution of $E_{i}$ is non-unique and can be any value within its domain. For convenience, we choose $E_{i}^{(\mu_i)}=0$ in this case. The similar case holds for subproblems in (\ref{sub3}), (\ref{sub4}) and (\ref{sub5}) if $\beta_E\mu_{\bar{\imath}}=\mu_i$, $\gamma_i\alpha^{G}_i= \mu_i$ and $\beta_B\lambda_{\bar{\imath}}=\lambda_{i}$, respectively. In these cases, $G_{i}^{(\mu_i)}=0$, $e_{i}^{(\mu_i)}=0$ and $w_{i}^{(\mu_i)}=0, i \in \{1,2\}$, are chosen as the solution for simplicity.}
{  Also note that the solutions in (\ref{result1}) - (\ref{result5}) are only for obtaining the dual function $g(\{\mu_i\},\{\lambda_i\})$ under any given ${\mu_i}$ and ${\lambda_i}$ to solve the dual problem (P1.1-D), while they may not be the optimal solution to the original(primal) problem (P1) duo to their non-uniqueness.}

 With the results in (\ref{result1}) $-$ (\ref{result5}), we have obtained the dual function $g(\{\mu_i\},\{\lambda_i\})$ with given $\{\mu_i\}$ and $\{\lambda_i\}$ satisfying (\ref{eqn:bounded:below}). Next, we maximize $g(\{\mu_i\},\{\lambda_i\})$ over $\{\mu_i\}$ and $\{\lambda_i\}$ to solve (P1.1-D). Since $g(\{\mu_i\},\{\lambda_i\})$ is convex but in general not differentiable, subgradient based algorithms such as the ellipsoid method \cite{BoydConvOptII} can be applied to solve (P1.1-D), where the subgradients of $g(\{\mu_i\},\{\lambda_i\})$ for $\mu_i$ and $\lambda_i$ are given by  $-\sum_{k\in\mathcal{K}_i}\frac{b_k^{(\mu_i,\lambda_i)}N_0}{g_k}\bigg(2^{\frac{r_k}{b_k^{(\mu_i,\lambda_i)}}}-1\bigg)-P_{c,i} +E_i^{(\mu_i)}$ and $-\sum_{k\in \mathcal{K}_1\cup\mathcal{K}_2}b_k^{(\mu_i, \lambda_i)}+W_i,~i \in \{1,2\}$, respectively. As a result, we can obtain the optimal dual solution as $\{\mu_i^{{\star}}\}$ and $\{\lambda_i^{{\star}}\}$.  Accordingly, the corresponding $\{b_{k}^{(\mu_i^{\star},\lambda_i^{\star})}\}$ becomes the optimal bandwidth allocation solution for (P1.1) and thus (P1), given by $\{b_{k}^{{\star}}\}$. Substituting the obtained $\{b_{k}^{{\star}}\}$ into (\ref{energy:nocoop}), the optimal power allocation solution for (P1) is thus obtained as $\{p_{k}^{{\star}}\}$.

However, it is worth noting that the other optimal optimization variables for (P1), given by $\{E_i^{\star}\},\{G_i^{{\star}}\},\{e_i^{\star}\}$ and $\{w_i^{\star}\}$, cannot be directly obtained from (\ref{result2})$-$(\ref{result5}), since the solutions in (\ref{result2})$-$(\ref{result5}) are in general non-unique. Nevertheless, it can be shown that the optimal solution of (P1) is always attained when the inequality constraints in (\ref{P1:cons1}) and (\ref{P1:cons3}) are tight. Thus, we have
\begin{align}
\sum_{k\in\mathcal{K}_i}p_{k}^{\star}+P_{c,i}=& E_i^{\star}+G_i^{\star}+\beta_{E} e_{\bar{\imath}}^{\star}-e_{i}^{\star} ,~ i \in \{1,2\}\label{P1:cons1:solution},\\
\sum_{k\in\mathcal{K}_i}b_k^{\star}=& W_i+\beta_Bw_{\bar{\imath}}^{\star}-w_i^{\star} ,~ i \in \{1,2\}\label{P1:cons3:solution}.
\end{align}
From (\ref{P1:cons3:solution}) and using the fact that $w_1^{\star}$ and $w_2^{\star}$ should not be positive at the same time, the optimal solution of $\{w_i^{\star}\}$ can be obtained.
Last, the optimal optimization variables $\{E_i^{\star}\},\{G_i^{\star}\}$ and $\{e_i^{\star}\}$ can be obtained by solving the LP in (P3). Therefore, Proposition 4.1 is proved.

{ \section{Proof of Lemma \ref{proposition:5.1}}\label{appendix:proof_lemma_5.1}
We prove Lemma \ref{proposition:5.1} by considering two cases. First, consider $\mv x^{\mathrm{ ex}}$ with $\sum_{k\in\mathcal{K}_i}p^{(\bm{x}^{\mathrm{ex}})}_k+P_{c,i}-\beta_E e_{\bar{\imath}}+e_{i}\neq\bar{E}_i$. In this case, it is evident from (\ref{P4dual:nocoop}) and (\ref{waterlevel}) that the optimal solutions $\mu_i^{(\mv x^{\mathrm {ex}})}$ and $\lambda_i^{(\mv x^{\mathrm {ex}})}$ are both unique, since  the bandwidth water-level $\nu_i$ is always unique.  According to Theorem 1 in \cite{Horst1984}, the left-partial derivative is equal to the right-partial derivative. Hence,  $C_i(\mv x^{\mathrm{ ex}})$ is differentiable with the partial derivatives given in Lemma \ref{proposition:5.1}. Therefore, (\ref{1storder}) follows directly based on the first order approximation of $\bar{C}_i(\mv{x}^{\mathrm{ex}}+\Delta \mv{x}^{\mathrm{ex}})$.

Next, consider $\mv x^{\mathrm {ex}}$ with $\sum_{k\in\mathcal{K}_i}p^{(\bm{x}^{\mathrm{ex}})}_k+P_{c,i}-\beta_E e_{\bar{\imath}}+e_{i}=\bar{E}_i$. In this case, the optimal dual solution of $\mu_i^{\bm{x}^{\mathrm {ex}}}$ in (\ref{P4dual:nocoop}) is not unique. More specifically, it can be shown that $\mu_i^{\mathrm {ex}}$ can be any real number between $\alpha_i^E$ to $\alpha_i^G$. As  a result, $C_i(\mv x^{\mathrm {ex}})$ is not differentiable in such point. However, it follows from \cite{Horst1984} that  the left- and right-hand derivatives of $C_i(\mv x^{\mathrm {ex}})$ with respect to $e_1,~e_2,~w_1$ and $w_2$ still exist, which can be given as
\begin{align}
\frac{\partial \bar{C}_i(\mv{x}^{\mathrm{ex}})}{\partial e_i^+}=\alpha_i^{G}&,~\frac{\partial \bar{C}_i(\mv{x}^{\mathrm{ex}})}{\partial e_i^-}=\alpha_i^{E}\nonumber\\
\frac{\partial \bar{C}_i(\mv{x}^{\mathrm{ex}})}{\partial e_{\bar{\imath}}^+}=-\beta_E\alpha_i^{E}&,~\frac{\partial \bar{C}_i(\mv{x}^{\mathrm{ex}})}{\partial e_{\bar{\imath}}^-}=-\beta_E\alpha_i^{G}\nonumber\\
\frac{\partial \bar{C}_i(\mv{x}^{\mathrm{ex}})}{\partial w_i^+}=\alpha_i^{G}\nu_i^{(\bm{x}^{\mathrm{ex}})}&,~\frac{\partial \bar{C}_i(\mv{x}^{\mathrm{ex}})}{\partial w_i^-}=\alpha_i^{E}\nu_i^{(\bm{x}^{\mathrm{ex}})}\nonumber\\
\frac{\partial \bar{C}_i(\mv{x}^{\mathrm{ex}})}{\partial w_{\bar{\imath}}^+}=-\alpha_i^{E}\nu_i^{(\bm{x}^{\mathrm{ex}})}&,~\frac{\partial \bar{C}_i(\mv{x}^{\mathrm{ex}})}{\partial w_{\bar{\imath}}^-}=-\alpha_i^{G}\nu_i^{(\bm{x}^{\mathrm{ex}})}.\nonumber
\end{align}
 By replacing the partial derivatives in (\ref{partialderivatives}) as the corresponding left- or right-hand derivatives, (\ref{1storder}) also follows from the first order approximation of $\bar{C}_i(\mv{x}^{\mathrm{ex}}+\Delta \mv{x}^{\mathrm{ex}})$. It is worth noting that, in Lemma \ref{proposition:5.1}, we do not introduce the left- and right-hand derivatives for notational convenience.

Therefore, Lemma 5.1 is proved.}
\section{Proof of Proposition \ref{theorem1}}\label{appendix:proof2}

Since the shared energy from BS 1 to BS 2 and that from BS 2 to BS 1 cannot be zero at the same time, i.e.,  $e_1\cdot e_2=0$, there exist three possible cases for the shared energy between the two BSs, which are (a) $e_1=e_2=0$, (b) $e_1>0, e_2=0$, and (c) $e_1=0, e_2>0$. Similarly, there are three possible cases for the shared bandwidth between the two BSs, i.e., (a) $w_1=w_2=0$, (b) $w_1>0, w_2=0$, and (c) $w_1=0, w_2>0$. As a result, by combining the above energy and spectrum cooperation, there are nine cases for the shared energy and the shared bandwidth. Therefore, we prove this proposition by enumerating the  nine possible cases. In the following, we consider the case of $e_1=e_2=w_1=w_2=0$ and show that in this case, $\mv{x}^{\mathrm{ex}}$ attains the Pareto optimality if and only if $\lambda_1^{(\mv{x}^{\mathrm{ex}})}/\mu_1^{(\mv{x}^{\mathrm{ex}})}\leq\lambda_2^{(\mv{x}^{\mathrm{ex}})}/(\mu_2^{(\mv{x}^{\mathrm{ex}})}\beta_E)$ and $\lambda_2^{(\mv{x}^{\mathrm{ex}})}/\mu_2^{(\mv{x}^{\mathrm{ex}})}\leq\lambda_1^{(\mv{x}^{\mathrm{ex}})}/(\mu_1^{(\mv{x}^{\mathrm{ex}})}\beta_E)$. We prove the ``only if'' and  ``if'' parts, respectively.

First, we show the necessary part by contradiction. Suppose that there exists an inter-system energy and bandwidth cooperation vector $\bar{\mv{x}}^{\mathrm{ex}}$ with $ \bar e_1=\bar e_2=\bar w_1=\bar w_2=0$ attains the Pareto optimality, where $\lambda_1^{(\bar {\mv{x}}^{\mathrm{ex}})}/\mu_1^{(\bar{\mv{x}}^{\mathrm{ex}})}>\lambda_2^{(\bar{\mv{x}}^{\mathrm{ex}})}/(\mu_2^{(\bar{\mv{x}}^{\mathrm{ex}})}\beta_E)$ or $\lambda_2^{(\bar{\mv{x}}^{\mathrm{ex}})}/\mu_2^{(\bar{\mv{x}}^{\mathrm{ex}})}>\lambda_1^{(\bar{\mv{x}}^{\mathrm{ex}})}/(\mu_1^{(\bar{\mv{x}}^{\mathrm{ex}})}\beta_E)$. If $\lambda_1^{(\bar {\mv{x}}^{\mathrm{ex}})}/\mu_1^{(\bar{\mv{x}}^{\mathrm{ex}})}>\lambda_2^{(\bar{\mv{x}}^{\mathrm{ex}})}/(\mu_2^{(\bar{\mv{x}}^{\mathrm{ex}})}\beta_E)$, then we can construct a new inter-system energy and bandwidth cooperation vector as
\begin{align}\label{neq:x_ex}
\tilde{\mv{x}}^{\mathrm{ex}} = \bar{\mv{x}}^{\mathrm{ex}} + \Delta {\mv{x}}^{\mathrm{ex}},
\end{align}
with $\Delta{\mv{x}}^{\mathrm{ex}} = (\Delta e_1, \Delta e_2, \Delta w_1, \Delta w_2)^T$, where $\Delta e_2= \Delta w_1=0$, while $\Delta e_1 > 0$ and $\Delta w_2>0$ are sufficiently small and satisfy that $\mu_1^{(\bar{\mv{x}}^{\mathrm{ex}})}/\lambda_1^{(\bar{\mv{x}}^{\mathrm{ex}})}<\Delta w_2/\Delta e_1<\beta_E\mu_2^{(\bar{\mv{x}}^{\mathrm{ex}})}/\lambda_2^{(\bar{\mv{x}}^{\mathrm{ex}})}.$ In this case, it can be shown from (\ref{1storder}) that
\begin{align}\label{eqn:76}
    \left[
      \begin{array}{c}
        \Delta C_1 \\
        \Delta C_2 \\
      \end{array}
    \right]&=\left[
      \begin{array}{c}
        \bar{C}_1(\tilde{\mv{x}}^{\mathrm{ex}})-\bar{C}_1(\bar{\mv{x}}^{\mathrm{ex}}) \\
        \bar{C}_2(\tilde{\mv{x}}^{\mathrm{ex}})-\bar{C}_2(\bar{\mv{x}}^{\mathrm{ex}}) \\
      \end{array}
    \right]\nonumber\\
&=
\left[
\begin{array}{c}
\mu_1^{(\bar{\mv{x}}^{\mathrm{ex}})}\Delta e_1-\lambda_1^{(\bar{\mv{x}}^{\mathrm{ex}})}\Delta w_2  \\
-\beta_E\mu_2^{(\bar{\mv{x}}^{\mathrm{ex}})}\Delta e_1+\lambda_2^{(\bar{\mv{x}}^{\mathrm{ex}})}\Delta w_2
\end{array}
\right] < \mv{0},
\end{align}
where the inequality is component-wise. In other words, we have found a new inter-system energy and bandwidth cooperation vector to achieve lower energy costs for both BSs. As a result, $\bar{\mv{x}}^{\mathrm{ex}}$ with $ \bar e_1=\bar e_2=\bar w_1=\bar w_2=0$ does not achieve the Pareto optimality. On the other hand, if $\lambda_2^{(\bar{\mv{x}}^{\mathrm{ex}})}/\mu_2^{(\bar{\mv{x}}^{\mathrm{ex}})}>\lambda_1^{(\bar{\mv{x}}^{\mathrm{ex}})}/(\mu_1^{(\bar{\mv{x}}^{\mathrm{ex}})}\beta_E)$, then we can also construct  a new inter-system energy and bandwidth cooperation vector as in (\ref{neq:x_ex}), where $\Delta e_1= \Delta w_2=0$, while $\Delta e_2 > 0$ and $\Delta w_1>0$ are sufficiently small and satisfy that $\mu_2^{(\bar{\mv{x}}^{\mathrm{ex}})}/\lambda_2^{(\bar{\mv{x}}^{\mathrm{ex}})}<\Delta w_1/\Delta e_2<\beta_E\mu_1^{(\bar{\mv{x}}^{\mathrm{ex}})}/\lambda_1^{(\bar{\mv{x}}^{\mathrm{ex}})}.$ Then, it can be shown that under this choice, the energy costs of the two BSs can be decreased at the same time. By combining the results for the two cases of $\lambda_1^{(\bar {\mv{x}}^{\mathrm{ex}})}/\mu_1^{(\bar{\mv{x}}^{\mathrm{ex}})}>\lambda_2^{(\bar{\mv{x}}^{\mathrm{ex}})}/(\mu_2^{(\bar{\mv{x}}^{\mathrm{ex}})}\beta_E)$ and $\lambda_2^{(\bar{\mv{x}}^{\mathrm{ex}})}/\mu_2^{(\bar{\mv{x}}^{\mathrm{ex}})}>\lambda_1^{(\bar{\mv{x}}^{\mathrm{ex}})}/(\mu_1^{(\bar{\mv{x}}^{\mathrm{ex}})}\beta_E)$, a contradiction is induced. As a result, the presumption cannot be true. Accordingly, the necessary part is proved.

Second, for the sufficient part, we can also show it by contradiction. Suppose that for an inter-system energy and bandwidth cooperation vector $\bar{\mv{x}}^{\mathrm{ex}}$ with $ \bar e_1=\bar e_2=\bar w_1=\bar w_2=0$ satisfying $\lambda_1^{(\bar{\mv{x}}^{\mathrm{ex}})}/\mu_1^{(\bar{\mv{x}}^{\mathrm{ex}})}\leq\lambda_2^{(\bar{\mv{x}}^{\mathrm{ex}})}/(\mu_2^{(\bar{\mv{x}}^{\mathrm{ex}})}\beta_E)$ and $\lambda_2^{(\bar{\mv{x}}^{\mathrm{ex}})}/\mu_2^{(\bar{\mv{x}}^{\mathrm{ex}})}\leq\lambda_1^{(\bar{\mv{x}}^{\mathrm{ex}})}/(\mu_1^{(\bar{\mv{x}}^{\mathrm{ex}})}\beta_E)$, but does not achieve the Pareto optimality. This case implies that the two BSs can exchange energy and bandwidth to decrease energy cost of both at the same time. In other words, there must exist a new vector $\tilde{\mv{x}}^{\mathrm{ex}}$ given in (\ref {neq:x_ex}) satisfying either $\Delta e_2= \Delta w_1=0$, $\Delta e_1 > 0,\Delta w_2>0$ or $\Delta e_1= \Delta w_2=0$, $\Delta e_2 > 0,\Delta w_1>0$ such that $\bar C_1(\tilde{\mv{ x}}^{\mathrm{ex}}) <  \bar C_1(\bar{\mv {x}}^{\mathrm{ex}})$ and $\bar C_2(\tilde {\mv{x}}^{\mathrm{ex}}) <  \bar C_2(\bar {\mv{x}}^{\mathrm{ex}})$. If $\Delta e_2= \Delta w_1=0$, $\Delta e_1 > 0,\Delta w_2>0$, then it can be shown from (\ref{eqn:76})
that $\lambda_1^{(\bar{\mv{x}}^{\mathrm{ex}})}/\mu_1^{(\bar{\mv{x}}^{\mathrm{ex}})}<\lambda_2^{(\bar{\mv{x}}^{\mathrm{ex}})}/(\mu_2^{(\bar{\mv{x}}^{\mathrm{ex}})}\beta_E)$ must hold, whereas if $\Delta e_1= \Delta w_2=0$, $\Delta e_2 > 0,\Delta w_1>0$, then we have $\lambda_2^{(\bar{\mv{x}}^{\mathrm{ex}})}/\mu_2^{(\bar{\mv{x}}^{\mathrm{ex}})}<\lambda_1^{(\bar{\mv{x}}^{\mathrm{ex}})}/(\mu_1^{(\bar{\mv{x}}^{\mathrm{ex}})}\beta_E)$. As a result, we have a contradiction here and thus the presumption cannot be true. Therefore, the sufficient  part is proved.

By combing the two parts, we have verified the proposition in the case of $e_1=e_2=w_1=w_2=0$.

Next,  the other eight cases remain to be proved in order to complete the proof of this proposition. Since the proof for these cases can follow the same contradiction  procedure as the case of $e_1=e_2=w_1=w_2=0$, we omit the details here for brevity. By combining the proof for the nine cases, this proposition is verified.

\bibliographystyle{IEEEtran}
\bibliography{Yinghao_OQE_ref}

\end{document}